\documentclass{amsart}
\usepackage[left=1in,right=1in,top=1.2in,bottom=1in]{geometry}
\usepackage[utf8]{inputenc}
\usepackage{amsmath,amssymb}
\usepackage{comment}
\usepackage{xcolor}
\usepackage{url}
\usepackage[colorlinks=true,allcolors=blue,backref=page]{hyperref}
\usepackage[noabbrev,capitalize,nameinlink]{cleveref}
\title[Down-up walk on independent sets]{Optimal mixing of the down-up walk on independent sets of a given size}
\author{Vishesh Jain}
\author{Marcus Michelen}
\address{University of Illinois Chicago}
\email{visheshj@uic,edu, michelen@uic.edu}
\author{Huy Tuan Pham}
\author{Thuy-Duong Vuong}
\address{Stanford University}
\email{huypham@stanford.edu, tdvuong@stanford.edu}
\renewcommand{\P}{\mathbb{P}}
\newcommand{\Var}{\operatorname{Var}}

\newcommand{\Ent}{\operatorname{Ent}}
\newcommand{\E}{\mathbb{E}}
\newcommand{\Z}{\mathbb{Z}}
\newcommand{\R}{\mathbb{R}}
\newcommand{\C}{\mathbb{C}}
\newcommand{\N}{\mathbb{N}}

\newcommand{\cC}{\mathcal{C}}
\newcommand{\cG}{\mathcal{G}}
\newcommand{\cE}{\mathcal{E}}
\newcommand{\cI}{\mathcal{I}}

\newcommand{\abs}[1]{|#1|
}
\newcommand{\mc}{\mathcal}
\newcommand{\blambda}{\boldsymbol{\lambda}}
\newcommand{\eps}{\varepsilon}
\newcommand{\ubar}{\overline{u}}
\newcommand{\ibar}{\overline{i}}

	\newtheorem{theorem}{Theorem}
	\newtheorem{definition}[theorem]{Definition}

	\newtheorem{lemma}[theorem]{Lemma}
	
	\newtheorem{corollary}[theorem]{Corollary}

	\newtheorem{fact}[theorem]{Fact}
	
	\newtheorem{prop}[theorem]{Proposition}

 \newtheorem*{claim*}{Claim}

\begin{document}

\maketitle

\begin{abstract}
Let $G$ be a graph on $n$ vertices of maximum degree $\Delta$. We show that, for any $\delta > 0$, the down-up walk on independent sets of size $k \leq (1-\delta)\alpha_c(\Delta)n$ mixes in time $O_{\Delta,\delta}(k\log{n})$, thereby resolving a conjecture of Davies and Perkins in an optimal form. Here, $\alpha_{c}(\Delta)n$ is the NP-hardness threshold for the problem of counting independent sets of a given size in a graph on $n$ vertices of maximum degree $\Delta$. Our mixing time has optimal dependence on $k,n$ for the entire range of $k$; previously, even polynomial mixing was not known. In fact, for $k = \Omega_{\Delta}(n)$ in this range, we establish a log-Sobolev inequality with optimal constant $\Omega_{\Delta,\delta}(1/n)$. 

At the heart of our proof are three new ingredients, which may be of independent interest. The first is a method for lifting $\ell_\infty$-independence from a suitable distribution on the discrete cube---in this case, the hard-core model---to the slice by proving stability of an Edgeworth expansion using a multivariate zero-free region for the base distribution. The second is a generalization of the Lee-Yau induction to prove log-Sobolev inequalities for distributions on the slice with considerably less symmetry than the uniform distribution. The third is a sharp decomposition-type result which provides a lossless comparison between the Dirichlet form of the original Markov chain and that of the so-called projected chain in the presence of a contractive coupling.
\end{abstract}

\section{Introduction}

Counting and sampling independent sets of a given size in a graph are intensely studied, fundamental computational tasks in a number of areas.  Specific examples include, for instance, the evaluation of the permanent of $0/1$-matrices (which is \#P-complete by a celebrated result of Valiant \cite{valiant1979complexity}); counting perfect matchings of planar graphs (which can be accomplished exactly in polynomial time by the classical FKT algorithm \cite{kasteleyn1967graph, temperley1961dimer}); and sampling stationary configurations of a conserved hard-core lattice gas on a graph with a prescribed number of particles \cite{bubley1997path}. Since these computational tasks are, in general, $\#$P-hard, the focus has been on algorithms for approximate counting and sampling. 

Given a graph $G$, let $\mathcal{I}_k(G)$ denote the set of independent sets of $G$ of size $k$, and let $\mu_k(G)$ denote the uniform distribution on $\mathcal{I}_k(G)$. For approximately sampling from $\mu_k(G)$ on general graphs $G$, the following extremely simple and natural Markov chain, popularly known as the down-up walk, has been investigated by various authors, going back essentially\footnote{The original work of Metropolis et al.~\cite{metropolis1953equation} considers a continuous analogue of the hard-core model known as the hard-sphere gas and worked with a generalization of the down-up walk: they introduce an additional parameter $\alpha$ and only allow ``down-up'' steps that move a point by at most $\alpha$ in the $\ell_\infty$ distance.  On the torus with $\alpha$ equalling the $\ell_\infty$ diameter, this is precisely the down-up walk.} to the original work of Metropolis, Rosenbluth, Rosenbluth, Teller and Teller \cite{metropolis1953equation} that introduced the Markov Chain Monte Carlo method: given the current state $I_{t} \in \mathcal{I}_k(G)$ at time $t$, independently select uniformly random vertices $u \in I_t ,v \in V(G)$, and let $I' = (I_t \setminus u) \cup v$. If $I' \in \mathcal{I}_k(G)$, then $I_{t+1} = I'$; if not, then $I_{t+1} = I_t$.  

In their seminal paper on path-coupling, Bubley and Dyer \cite{bubley1997path} showed that for all graphs on $n$ vertices with maximum degree $\Delta$, and for all $k \leq (1-\delta)n/(2\Delta+2)$, the down-up walk mixes in time $O_{\delta}(k\log{n})$. A few years ago, in one of the early works on the application of high-dimensional expander (HDX) techniques to MCMC, Alev and Lau \cite{alev2020improved} showed that the down-up walk mixes in time $O(k^{3}\log(n/k))$ for $k \leq n/(\Delta + |\lambda_{\min}(A_G)|)$, where $A_G$ denotes the adjacency matrix of the graph $G$; as a restriction on $k$ parameterized solely by the maximum degree, this translates to $k \leq n/(2\Delta)$ since for a graph of maximum degree $\Delta$, $|\lambda_{\min}(A_G)|$ could be as large as $\Delta$. 

The computational complexity of the problem of approximately sampling from $\mu_k(G)$ (or essentially equivalently, approximating $|\mathcal{I}_k(G)|$) on $n$-vertex graphs $G$ of maximum degree $\Delta$ was recently investigated systematically by Davies and Perkins \cite{DP21}. They showed that there is an explicit function \footnote{see \cref{sec:spectral-independence} for an interpretation of this function. Here, we only note that $\alpha_c(\Delta) = \frac{(1+o_{\Delta}(1))e}{(1+e)\Delta}$.}
\[\alpha_c(\Delta) := \frac{(\Delta-1)^{\Delta-1}}{(\Delta-2)^{\Delta} + (\Delta+1)(\Delta-1)^{\Delta-1}}\]
such that for any $\alpha < \alpha_{c}(\Delta)$, there is an FPRAS for $|\mathcal{I}_k(G)|$ for all $k \leq \alpha n$ and conversely, no FPRAS exists for $k \geq \alpha n$, $\alpha > \alpha_{c}(\Delta)$, unless NP = RP; by standard reductions, the same holds for approximately sampling from $\mu_k(G)$.  

Davies and Perkins conjectured (\cite[Conjecture~5]{DP21}) that in the above setting, the down-up walk mixes in polynomial time for all $k \leq \alpha n$ with $\alpha < \alpha_{c}(\Delta)$. Our main result resolves this conjecture in the affirmative in essentially the strongest-possible form by providing a bound on the mixing time which has optimal dependence on $k$ and $n$.

\begin{theorem}
    \label{thm:main}
    Let $\Delta \geq 3$ and $\delta \in (0,1)$. For a graph $G = (V,E)$ on $n$ vertices of maximum degree at most $\Delta$ and an integer $1 \leq k \leq (1-\delta)\alpha_c(\Delta)n$, the down-up walk on independent sets of size $k$
    has $\varepsilon$-mixing time $O_{\delta,\Delta}(k\log(n/\varepsilon))$.  \footnote{Recall that the $\varepsilon$-mixing time of a Markov chain with transition matrix $P$ and stationary distribution $\mu$ on state space $\Omega$ is defined to be $\tau_{\text{mix}}(\varepsilon) = \max_{\nu}\min\{t \geq 0: \text{TV}(\nu P^{t}, \mu)\leq \varepsilon\}$, where $\text{TV}$ denotes the total variation distance between probability distributions and the max ranges over all probability distributions $\nu$ on $\Omega$.}
\end{theorem}

Given the work of Bubley and Dyer \cite{bubley1997path}, it suffices to prove \cref{thm:main} for $k = \Omega_{\Delta}(n)$. In this case, we establish the stronger result that the down-up walk satisfies a log-Sobolev inequality with constant $\Omega_{\Delta,\delta}(1/n)$. Apart from immediately implying the above result on mixing times, this has various additional consequences for the stationary measure $\mu_k$ such as sub-Gaussian concentration of Lipschitz functions, transport-entropy inequalities, and hypercontractivity with respect to the corresponding semi-group (see \cite{bobkov2006modified}).  

\begin{theorem}
    \label{thm:main-lsi}
     Let $\Delta \geq 3$ and $\delta \in (0,1)$. For a graph $G = (V,E)$ on $n$ vertices of maximum degree at most $\Delta$ and an integer $c n \leq k\leq (1-\delta)\alpha_c(\Delta)n$, the down-up walk on independent sets of size $k$ satisfies a log-Sobolev inequality with constant $\Omega_{c,\delta,\Delta}(1/n)$. 
\end{theorem}

While log-Sobolev and modified log-Sobolev inequalities for canonical walks with respect to the uniform distribution on the (multi)slice have been established in numerous works (e.g., \cite{diaconis1981generating, diaconis1996logarithmic, lee1998logarithmic, filmus2022log, salez2021sharp} and the references therein), to the best of our knowledge, \cref{thm:main-lsi} is the first instance of an asymptotically optimal log-Sobolev inequality for a highly non-symmetric natural distribution supported on the Boolean slice.

\subsection{Overview of techniques} We conclude with a brief overview of our techniques.\\ 
\paragraph{\bf Spectral independence for $\mu_k$} Let $G$ be a graph on $n$ vertices with maximum degree $\Delta \geq 3$. Since $\alpha_{c}(\Delta) < \frac{1}{\Delta+1}$, it follows that the condition $k \leq (1-\delta)\alpha_{c}(\Delta)n$ is closed under pinning vertices of $G$ to belong to the independent set. Hence, one might hope to prove rapid mixing of the down-up walk by establishing $O_{\delta,\Delta}(1)$-spectral independence for $\mu_k$ (viewed as a distribution on $\binom{[n]}{k}$). For proving spectral independence of a distribution for the purposes of establishing rapid mixing (as opposed to optimal mixing) of the down-up walk, there are primarily three techniques: Oppenheim's trickle down theorem \cite{oppenheim2018local}, zero-free regions of multivariate generating polynomial \cite{chen2022spectral}, and tree recursions using a suitable potential function (e.g.~\cite{chen2023rapid}). The first of these was used by Alev and Lau \cite{alev2020improved} and works until $k \leq n/(2\Delta)$. While the latter two techniques have been successfully used in the case of the hard-core model in the uniqueness region, it is unclear how to adapt them to our cardinality-constrained model. A salient challenge here is that any potential approach must be able to ``witness'' the threshold $\alpha_{c}(\Delta)n$, which is most naturally interpreted in terms of the uniqueness threshold $\lambda_{c}(\Delta)$ of the hard-core model, thereby ruling out several purely ``slice-based'' approaches.

In order to be able to witness $\alpha_{c}(\Delta)n$, we view $\mu_k$ as being obtained by rejection sampling from a hard-core model at a suitable activity $\lambda$ in the uniqueness region; we note that this is precisely the efficient sampling algorithm for $\mu_k$, $k\leq (1-\delta)\alpha_{c}(\Delta)n$ provided by Davies and Perkins \cite{DP21}. Adopting this viewpoint, spectral independence of $\mu_k$ (in the strengthened form, known as $\ell_\infty$-independence) follows from the known $\ell_\infty$-independence of the hard-core model in the uniqueness region, provided we have a fine understanding of $\P_{\lambda,G}(|I|=k)$ -- the probability that an independent set drawn from the hard-core model at activity $\lambda$ has size $k$. Specifically, we need to show that $\P_{\lambda,G}(|I|=k) = \Omega_{\delta,\Delta}(n^{-1/2})$ (this follows immediately from a local central limit theorem (LCLT) for the hard-core model established in \cite{jain2022approximate}) and critically, that if $G'$ is a graph obtained from $G$ by adding or removing $O_{\Delta}(1)$ vertices, then $|\P_{\lambda,G'}(|I|=k) - \P_{\lambda,G}(|I| = k)| = O_{\delta,\Delta}(n^{-3/2})$; in contrast to the lower bound on the probability, which is a consequence of the Gaussian behavior of $|I|$ (as encapsulated by the LCLT), the second point requires us to show that the deviations from Gaussianity are optimally stable under perturbations of the underlying graph $G$. We accomplish this by establishing an Edgeworth expansion for the probability to within arbitrarily small polynomial error (this requires estimates on high Fourier coefficients of $|I|$ from \cite{jain2022approximate}) and then showing that the coefficients of the higher order terms in the expansion---which are polynomials in the  cumulants of $|I|$---are stable under perturbations of $G$, by combining the presence of the zero-free region for the (complex) multivariate independence polynomial from \cite{PR19} with an application of Cauchy's integral formula to bound the magnitude of various derivatives. 
The use of the (univariate) zero-free region is a key ingredient in the local central limit theorem in \cite{jain2022approximate},
and follows a line of literature proving central limit theorems for spin systems from zero-free regions \cite{iagolnitzer1979lee,lebowitz2016central}  as well as more general central limit theorems from zero-free regions \cite{ghosh2017multivariate,lebowitz2016central,michelen2019central,michelen2019clt}.  Our contribution shows that for a spin system in the presence of a \emph{multivariate} zero-free region, one can in fact write an asymptotic expansion for probabilities such as $\P_{\lambda,G}(|I| = k)$ for $k$ near the mean.  
\\
\paragraph{\bf Log-Sobolev inequality} Chen, Liu, and Vigoda \cite{CLV20} showed how to leverage spectral independence to prove log-Sobolev inequalities for Markov random fields on bounded degree graphs, provided that all marginals of the distribution are uniformly lower bounded under arbitrary pinnings. While our distribution may be viewed as a spin system on a bounded degree graph, the other properties fail to hold: $\mu_k$ is neither a Markov random field nor has the bounded marginal property; indeed, conditioning on $k-1$ vertices to be in the independent set makes all the remaining marginals $O(1/n)$. Overcoming these challenges requires a number of innovations.

First, noting that the distributions obtained by pinning at most $k - \Omega_{\Delta}(k)$ vertices are marginally bounded, we can use the local-to-global/annealing machinery from the spectral independence/localization-schemes literature to reduce our task to proving a log-Sobolev inequality for $\mu_k$, $k \leq n/\Delta^8$ (say). Compared to the hard-core model at activity $\lambda \sim 1/\Delta^8$, this turns out to be a much harder task; indeed, even proving an asymptotically optimal log-Sobolev inequality for the uniform distribution on the Boolean slice $\binom{[n]}{k}$ is non-trivial and was only determined, after much work, by Lee and Yau \cite{lee1998logarithmic} (the uniform distribution on $\binom{[n]}{k}$ is the same as the uniform distribution over independent sets of size $k$ in the empty graph, hence a special case of our setup). In our setting, since there is no symmetry, the situation is much more complicated. The main part of our argument here is \cref{thm:lee-yau}, which proves an optimal log-Sobolev inequality assuming that the graph has linearly many connected components, each of which has size at most logarithmic in the number of vertices. 
Finally, to reduce to the setting of \cref{thm:lee-yau}, we use an average case ``annealing'' argument (cf.~\cite{chen2022localization}) building on a technique in \cite{CLV20}. 

For controlling the log-Sobolev constant in a graph with linearly many, at most logarithmically-sized components (\cref{thm:lee-yau}), we follow the overall inductive approach of Lee and Yau \cite{lee1998logarithmic} for bounding the log-Sobolev inequality on the slice.  This proceeds through the \emph{a priori} control of the log-Sobolev constant of the ``projected'' distribution on each small connected component combined with an inductive control of the log-Sobolev constant of the distribution restricted to the complement of the component, averaged out over all components. In our case, both the control of the projected distribution and the inductive control of the distribution restricted to the complement are highly non-trivial. 

In order to control the contribution from the projected distribution, we observe that using the LCLT from \cite{jain2022approximate}, the log-Sobolev constant of the projected distribution on each component can be bounded by direct comparison with the hardcore model on the component at the corresponding activity. Given this, it remains to upper bound the Dirichlet form of this projected distribution by appropriate terms in the Dirichlet form of the original distribution. In special cases, such as when the graph is empty, such a bound can be obtained in a straightforward manner using the convexity of the functionals appearing in the Dirichlet forms. However, convexity arguments are not available in general. In fact, proving a bound of this form is one of the key steps in Markov chain decomposition results (e.g.~\cite{jerrum2004elementary}); existing results of this form lose a factor depending on the spectral gap of the distribution, which in our case is far too large for the induction to work. Instead, we provide a novel, sharp decomposition result, which allows us to give an optimal comparison (up to constant factors) between the Dirichlet form of the projected chain and that of the original distribution. Our analysis makes use of Stein's method for Markov chains \cite{bresler2019stein, reinert2019approximating} and leverages the existence of a contractive coupling in the regime we operate in. 

The other main part of the argument is the inductive control of the log-Sobolev constant on the entire graph, which is again much more involved than the case of the uniform distribution on the slice. Instead of directly inducting on the bound, we simulate the inductive bound  via a stochastic process involving one effective parameter which is the analog of the occupancy ratio of the independent set along the process. This effective parameter is closely captured by a martingale with bounded differences and we control it by the martingale Bernstein inequality.

\subsection{Organization}
In \cref{sec:prelims}, we collect some preliminaries. \cref{sec:spectral-independence} is devoted to our first main ingredient, namely the proof that $\mu_k$ is spectrally independent for $k \leq (1-\delta)\alpha_c(\Delta)n$ (\cref{th:SI-slice}); combining this with standard techniques in the theory of spectral independence easily shows that the down-up walk on $\mathcal{I}_k(G)$ has optimal spectral gap $\Omega_{\delta,\Delta}(1/k)$ (\cref{thm:spectral-gap}). Obtaining a log-Sobolev inequality from our spectral independence result is considerably more involved; this is the content of \cref{sec:LSI,sec:lee-yau}.

\section{Preliminaries}
\label{sec:prelims}

In this section, we record some preliminaries for later use.\\ 

\paragraph{\bf Hard-core model} The hard-core model on a graph $G = (V,E)$ at activity $\lambda \in \R_{> 0}$ is the probability distribution on $\mathcal{I}(G)$, the independent sets of $G$, defined by
\[\mu_{G,\lambda}(I) = \frac{\lambda^{|I|}}{Z_G(\lambda)},\]
where $Z_G(\lambda) = \sum_{I \in \mathcal{I}(G)}\lambda^{|I|} $ is the independence polynomial of $G$. We will view the independence polynomial as a univariate polynomial of a complex-valued random variable $\lambda$.

For $\Delta \ge 3$, let 
 $$\lambda_c(\Delta) = \frac{(\Delta-1)^{\Delta-1}}{(\Delta-2)^{\Delta}};$$
 this is the uniqueness threshold for the hard-core model on the infinite $\Delta$-regular tree.  For $ 0\leq \lambda < \lambda_c(\Delta)$,  Weitz  gave an FPTAS for $Z_G(\lambda)$ on the class of graphs of maximum degree $\Delta$~\cite{weitz2006counting}.  Sly~\cite{sly2010computational}, Sly and Sun~\cite{sly2014counting}, and Galanis, {\v{S}}tefankovi{\v{c}}, and Vigoda~\cite{galanis2016inapproximability} complemented this by showing that for $\lambda > \lambda_c(\Delta)$, no FPRAS for $Z_G(\lambda)$ exists unless NP=RP.

As mentioned above, Davies and Perkins \cite{DP21} established a corresponding threshold for counting independent sets of a fixed size in bounded degree graphs, namely
\[ \alpha_c(\Delta) = \frac{\lambda_c(\Delta)}{1+(\Delta+1)\lambda_c(\Delta)} = \frac{(\Delta-1)^{\Delta-1}}{(\Delta-2)^{\Delta}+(\Delta+1)(\Delta-1)^{\Delta-1}};\]
this is the occupancy fraction (i.e.~the expected density of an independent set) for the hard-core model on the clique on $\Delta + 1$ vertices at the critical fugacity $\lambda_c(\Delta)$. Since a clique on $\Delta+1$ vertices minimizes the occupancy fraction in the class of graphs of maximum degree $\Delta$, this immediately implies the following. 

\begin{lemma}[see, e.g., \cite{DP21}] 
\label{lem:finding-activity}
Let $G$ be a graph on $n$ vertices with maximum degree $\Delta$. For any $\gamma > 0, \delta > 0$, there exist $\gamma' > 0, \delta' > 0$ depending only on $\gamma, \delta, \Delta$ such that for any $\gamma n \leq k \leq (1-\delta)\alpha_c(\Delta)n$, there exists $\gamma' \leq \lambda \leq (1-\delta')\lambda_c(\Delta)$ such that
\[\E_{I\sim \mu_{G,\lambda}}[|I|] = k\]  
\end{lemma}

\paragraph{\bf Functional inequalities} Let $P$ be the transition matrix of an ergodic, reversible Markov chain on a finite set $\Omega$, with (unique) stationary distribution $\mu$. The Dirichlet form of $P$ is defined, for $f,g \colon \Omega \to \R$, by
\[\mathcal{E}_P(f,g) := \frac{1}{2}\sum_{x,y \in \Omega}\mu(x)P(x,y)(f(x) - f(y))(g(x) - g(y)).\]

\begin{definition}
The spectral gap or Poincar\'e constant of $P$ is defined to be $\gamma$, where $\gamma$ is the largest value such that for every $f \colon \Omega \to \R$,
\[\gamma \Var_{\mu}{f} \leq \mathcal{E}_P(f,f).\]
The log-Sobolev constant of $P$ is defined to be the the largest value $\rho_{LS}$ such that for every $f \colon \Omega \to \R_{\geq 0}$,
\[\rho_{LS} \Ent_{\mu}{f} \leq \mathcal{E}_{P}(\sqrt{f},\sqrt{f}),\]
where $\Ent_{\mu}{f} = \E_{\mu}{f\log f} - \E_\mu{f} \log (\E_\mu{f})$.

\end{definition}

The following relationship between the $\varepsilon$-(total variation) mixing time of $P$, $\tau_{\text{mix}}(\varepsilon)$, and its Poincar\'e and log-Sobolev constants is standard (see, e.g.,~\cite{bobkov2006modified}):
\begin{align}
\label{eqn:mixing-time}
    \tau_{\text{mix}}(\varepsilon) &\leq \gamma^{-1}\log\left(\frac{1}{\varepsilon}\cdot \frac{1}{\min_{x\in \Omega}\mu(x)}\right),\\
    \tau_{\text{mix}}(\varepsilon) &\leq \rho_{LS}^{-1}\left(\log\log\left(\frac{1}{\min_{x\in \Omega}\mu(x)}\right) + \log\left(\frac{1}{2\varepsilon^{2}}\right)\right). \nonumber 
\end{align}

\paragraph{\bf Spectral independence and $\ell_\infty$-independence} Our proofs will make use of the notions of spectral independence \cite{anari2020spectral} and a strengthened version, sometimes referred to as $\ell_\infty$-independence. Here, and later, we use the following notation. For a distribution $\mu$ on subsets of $[n]$, $\P_{\mu}[i] = \P_{S \sim \mu}[i \in S]$ and $\P_{\mu}[\ibar] = \P_{S \sim \mu}[i \notin S]$. 
\begin{definition}[Influence matrix] \label{def:corr}
	Let $\mu$ be a probability distribution over $2^{[n]}$. 
Its (signed) pairwise influence matrix $M_{\mu} \in \R^{n\times n}$  is defined by
\[M_{\mu} (i,j) = \begin{cases} 0 &\text{ if } j=i, \\ \P_{\mu}[j \,|\, i] - \P_{\mu}[j \,|\, \ibar] &\text{ otherwise.}\end{cases}.\]
\end{definition}

\begin{definition}[Spectral Independence]
For $\eta \geq 0$, a distribution 
$\mu: 2^{[n]} \to \R_{\geq 0} $ is said to be $\eta$-spectrally independent (at the link $\emptyset$) if
\[\lambda_{\max}(M_{\mu}) \leq \eta.\] 
\end{definition}

Since the spectral radius of a matrix is bounded above by the $\ell_\infty \to \ell_\infty$ operator norm, it follows that spectral independence is an immediate consequence of the following. 

\begin{definition}[$\ell_\infty$-independence]
For $\eta \geq 0$, a distribution 
$\mu: 2^{[n]} \to \R_{\geq 0} $ is said to be $\eta$-$\ell_\infty$-independent (at the link $\emptyset$) if
\[\max_{i \in [n]}\sum_{j=1}^{n}|M_\mu(i,j)|\leq \eta.\] 
\end{definition}

\paragraph{\bf Comparison with the HDX down-up walk} The high-dimensional expander (HDX) down-up walk on $\mathcal{I}_k(G)$ is a slight variant of the down-up walk defined in the introduction, the difference being that instead of choosing a uniform $v \in V(G)$, we choose a uniform $v$ so that $I' = (I_t \setminus u) \cup v$ is a uniform element of $\mathcal{I}_k(G)$ containing $I_t \setminus u$. In other words, the HDX down-up walk is equivalent to the down-up walk, conditioned on $I' \in \mathcal{I}_k(G)$. Both walks are reversible with respect to the uniform distribution $\mu_k$ on $\mathcal{I}_k(G)$. Moreover, since in our results, $k \leq \alpha_c(\Delta)n \leq \frac{16n}{17(\Delta+1)}$, it follows that under the down-up walk in the introduction, the probability that $I' \in \mathcal{I}_k$ is at least $1/17$. 

Since the only successful transitions of the down-up walk are when $I' \in \mathcal{I}_k(G)$, it follows that mixing times of $P$ and $Q$ are within constant factors of each other, where $P$ is the down-up walk in the introduction and $Q$ is the HDX down-up walk. Similarly, the Dirichlet forms $\mathcal{E}_P(f,f)$ and $\mathcal{E}_Q(f,f)$, and consequently the spectral gap and the log-Sobolev constant, are within a constant factor of each other. Therefore, it suffices to prove \cref{thm:main,thm:main-lsi} for the HDX down-up walk.\\

\section{$\ell_\infty$-independence for $\mu_k$ and the spectral gap of the down-up walk}
\label{sec:spectral-independence}
Throughout this section, $\P_k$, $\E_k$ denote probabilities and expectations with respect to the measure $\mu_k$ and $\P_\lambda, \E_\lambda$ denotes probabilities and expectations with respect to the hard-core model $\mu_\lambda$ at activity $\lambda$; the underlying graph will be clear from context. For brevity, for $u \in V(G)$, we let $\P_k(u)$ (respectively, $\P_k(\ubar)$) denote the probability that an independent set sampled from $\mu_k$ contains $u$ (respectively, does not contain $u$) and similarly for $\P_\lambda(u), \P_\lambda(\ubar)$. The following is the main result of this section. 

\begin{theorem}[$\ell_\infty$-independence for $\mu_k$] \label{th:SI-slice}
    Let $G$ be a graph on $n$ vertices of maximum degree $\Delta$ and fix $\delta,\gamma > 0$.  Then for all $\gamma n \leq k \leq (1 - \delta)\alpha_c(\Delta)n$ and $u\in V(G)$ we have 
    $$
    \sum_{v \in V(G)} \left| \P_k(v\,|\,\ubar) - \P_k(v\,|\,u) \right| = O_{\Delta,\delta,\gamma}(1)\,.
    $$
\end{theorem}

The corresponding result for the hard-core model at activity $\lambda \leq (1-\delta)\lambda_c(\Delta)$ was established a few years ago (\cite{chen2023rapid}, see also \cite{anari2020spectral}).  

\begin{theorem}[$\ell_\infty$-independence of the hard-core model] \label{th:SI-hardcore}  
    Let $G$ be a graph on $n$ vertices of maximum degree $\Delta$ and fix $\delta > 0$.  Then for all  $\lambda \leq (1 -\delta)\lambda_c(\Delta)$ and $u \in V(G)$, we have $$
    \sum_{v \in V(G)} \left| \P_\lambda(v\,|\,\ubar) - \P_\lambda(v\,|\,u) \right| = O_{\Delta,\delta}(1)\,.$$
\end{theorem}

Our high-level strategy will be to deduce \cref{th:SI-slice} from \cref{th:SI-hardcore} by viewing $\mu_k$ as the conditional distribution $\mu_{\lambda}(\cdot \mid |I|=k)$, where $\Omega_{\gamma}(1) = \lambda \leq (1-\delta)\lambda_c(\Delta)$ is chosen to ensure that $\E_{\lambda}[|I|] = k$ (\cref{lem:finding-activity}). Adopting this point of view, we can rewrite
\begin{equation}
    \label{eqn:hardcore-to-slice}
    |\P_k(v\,|\,\ubar) - \P_k(v\,|\,u)| =  \left|\P_\lambda(v\,|\,\ubar)\frac{\P_\lambda(|I| = k \,|\, v \cap \ubar)}{\P_\lambda(|I| = k\,|\,\ubar)} - \P_\lambda(v\,|\,u)\frac{\P_\lambda(|I| = k \,|\, v \cap u)}{\P_\lambda(|I| = k\,|\,u)} \right|\,.
\end{equation}
In light of \cref{th:SI-hardcore}, it therefore suffices to show that $$\P_{\lambda}(|I| = k) = \Theta_{\Delta, \delta, \gamma}(n^{-1/2})$$ and
\begin{equation}
\label{eqn:prob-hitting}
\P_\lambda(|I| = k \,|\, \bullet) = \P_\lambda(|I| = k) + O_{\Delta,\delta,\gamma}(n^{-3/2})\,
\end{equation}
for each of the four quantities of this form appearing in \eqref{eqn:hardcore-to-slice}. The first of these equations was established in \cite{jain2022approximate}, as a consequence of a local central limit theorem for $|I|$ and the fact that $\Var_{\lambda}(|I|) = \Theta_{\Delta,\delta,\gamma}(n)$ (see \cite[Lemma~3.2]{jain2022approximate}). The heart of the matter is \eqref{eqn:prob-hitting}. Since the variance of $|I|$ is only $\Theta_{\Delta,\delta,\gamma}(n)$, a local central limit theorem by itself is too crude to establish a statement of this sort.

There are two main ingredients in the proof of \eqref{eqn:prob-hitting}: an Edgeworth expansion -- to an arbitrary polynomial error -- for $\P_{\lambda,G}(|I| = k)$; and the property that the cumulants of the random variable $|I|$ under the hard-core model at activity $\lambda$ are stable under small perturbations of the underlying graph $G$. 

Recall that the $j^{th}$ cumulant of a random variable $X$ is defined in terms of the coefficients of the cumulant generating function $K_X(t) := \log\E e^{tX}$ (when this expectation exists in a neighborhood of 0). In
particular, the $j^{th}$ cumulant is
\[\kappa_j(X) = \frac{d^j}{dt^j}K_X(t)|_{t=0}.\]
The first and second cumulants are the mean and variance respectively. 

\subsection{An Edgeworth expansion: the probability of hitting the mean}

To state the Edgeworth expansion, we first recall the Hermite polynomials $H_k(x)$, which are defined via $$H_k(x) = (-1)^k e^{x^2 / 2}\frac{d^k}{dx^k} e^{-x^2 / 2} $$
and satisfy 
\begin{equation}\label{eq:Hermite-identity}
    \frac{1}{\sqrt{2\pi}} \int_{\R} e^{-t^2/2} e^{-itx} (it)^k \,dt = H_k(x) e^{-x^2/2}\,.
\end{equation}

We now prove the Edgeworth expansion for the size of an independent set sampled from the hard-core model.

\begin{prop}\label{prop:edgeworth}
    Let $G$ be a graph on $n$ vertices of maximum degree $\Delta$ and fix $\delta,\gamma > 0$.  Let $X$ be the random variable giving the size of an independent set in $G$ sampled from the hard-core model at activity $\gamma \leq \lambda \leq (1 - \delta)\lambda_c(\Delta)$.   Set $\mu = \E X$,  $\sigma^2 = \Var(X)$ and $\beta_j = \kappa_j(X)/(j!\sigma^j)$.  Then, for each fixed $d \in \N$ and $a \in \R$ so that $\mu + a \in \Z$ we have 
        
    \begin{equation*}
        \P(X = \mu + a) = \frac{e^{-a^2/2\sigma^2}}{\sqrt{2\pi} \sigma} \left(1 + \sum_{r \geq 3} H_r(a/\sigma) \sum_{j_3,\ldots,j_\ell} \frac{\beta_3^{j_3}}{j_3!}\cdots \frac{\beta_\ell^{j_\ell}}{j_\ell!} \right) + O_{\Delta,\delta,\gamma,d}(n^{-d}) 
    \end{equation*}
    where the inner sum is over sequences of non-negative integers $j_3,j_4,\ldots,j_\ell$ so that $\sum_a a j_a = r$ and $j_3/2 + j_4 + j_5(3/2) + \cdots + j_\ell(\ell-2)/2 \leq d$.
\end{prop}
\begin{proof}
By Fourier inversion, we have 
\begin{equation*}
    \P(X = \mu + a) = \frac{1}{2\pi \sigma} \int_{-\pi \sigma}^{\pi \sigma} e^{-ita/\sigma} \E e^{it (X-\mu) /\sigma}\,dt\,.
\end{equation*}
By \cite[Lemma~3.5]{jain2022approximate}, there is a constant $c = c(\Delta,\delta,\gamma)$ so that \begin{equation*}
	| \E e^{it X/\sigma}| \leq \exp(-c n t^2/\sigma^2) 
\end{equation*}
for all $|t| \leq \pi\sigma$. Since $\sigma^2 = \Theta_{\Delta,\delta,\gamma}(n)$ (see \cite[Lemma~3.2]{jain2022approximate}), it follows that we can take $C = C(\Delta,\delta,\gamma,d)$ large enough so that we have \begin{equation}\label{eq:truncated-integral}
	\P(X = \mu + a) = \frac{1}{2\pi \sigma} \int_{|t| \leq C \log n} e^{-ita/\sigma} \E e^{i t (X-\mu) /\sigma}\,dt + O_{\Delta,\delta,\gamma,d}(n^{-d})\,.
\end{equation}

To proceed, we will need the following lemma, which will be proved in the next subsection. 

\begin{lemma}
\label{lem:cumulant-bound}
There exists a sufficiently small positive constant $\eps_{\mathrm{PR}} > 0$, depending only on $\Delta, \delta$, such that for each $d \in \N$ 
 $$|\kappa_d(X)| = O_{\Delta,\delta}(d!2^d\eps_{\mathrm{PR}}^{-d} n)\,$$
and moreover, for all $|t| \leq \eps_{\mathrm{PR}}/4$, we have $$\left|\log \frac{Z_G(\lambda e^t)}{Z_G(\lambda)} - \sum_{j = 1}^d \kappa_j(X) \frac{t^j}{j!}\right| = O_{\Delta,\delta}(|2t/\eps_{\mathrm{PR}}|^{d+1} n)\,. $$
\end{lemma}

From \cref{lem:cumulant-bound} and $\sigma^2 = \Theta_{\Delta,\delta,\gamma}(n)$, we can easily deduce the following.

\begin{lemma}
    For $t \leq C \log n$ we have 
    $$\E e^{it(X-\mu)/\sigma} = e^{-t^2/2}\left(1 + \sum_{r\geq 3} (it)^r \sum_{j_3,\ldots,j_\ell} \frac{\beta_3^{j_3}}{j_3!}\cdots \frac{\beta_\ell^{j_\ell}}{j_\ell!} \right)(1 + O_{\Delta,\delta,\gamma,d}(n^{-d}\log^{2d+2}n)) $$
    where the inner sum is over sequences of non-negative integers $j_3,j_4,\ldots,j_\ell$ so that $\sum_{a\geq 3} a j_a = r$ and $j_3/2 + j_4 + j_5(3/2) + \cdots + j_\ell(\ell-2)/2 \leq d$.
\end{lemma}
\begin{proof}[Proof of Lemma]
Apply \cref{lem:cumulant-bound} to bound \begin{align*}
        \E e^{it(X-\mu)/\sigma} &= e^{-it\mu/\sigma}\frac{Z_G(\lambda e^{it/\sigma})}{Z_G(\lambda)} = e^{-it\mu/\sigma}\exp\left(\log \frac{Z_G(\lambda e^{it/\sigma})}{Z_G(\lambda)}\right)\\
        &= e^{-it\mu/\sigma}\exp\left(\sum_{j=1}^{2d+1} \kappa_j(X) \frac{(it)^{j}}{\sigma^j j!} + O_{\Delta,\delta,\gamma,d}\left(\frac{\log^{2d+2}n}{n^{d}}\right)\right)\\
        &= e^{-it\mu/\sigma}\exp\left(it\mu/\sigma -t^2/2 + \sum_{j = 3}^{2d+1}\beta_j (it)^j + O_{\Delta,\delta,\gamma,d}(n^{-d}\log^{2d+2} n)  \right) \\
        &= e^{-t^2/2} \prod_{j = 3}^{2d+1} \exp\left(\beta_j (it)^j\right)\left(1 + O_{\Delta,\delta,\gamma,d}(n^{-d}\log^{2d+2} n)  \right)  \,.
     \end{align*}

    By \cite[Lemma~3.2]{jain2022approximate}, we have $\sigma^2 = \Theta_{\Delta,\delta,\gamma}(n)$; combining this with \cref{lem:cumulant-bound}, we have that \begin{equation}\label{eq:beta-bound}
        |\beta_j| = O_{\Delta,\delta,\gamma,j}(n^{1-j/2})\,.
    \end{equation}  
Finally, expanding the exponential terms in the above product and collecting terms completes the proof.
\end{proof}
Combining the preceding lemma with \eqref{eq:truncated-integral} shows $$\P(X = \mu + a) = \frac{1}{2\pi \sigma} \int_{|t| \leq C \log n} e^{-ita/\sigma} e^{-t^2/2}\left(1+ \sum_{r\geq 3} (it)^r \sum_{j_3,\ldots,j_\ell} \frac{\beta_3^{j_3}}{j_3!}\cdots \frac{\beta_\ell^{j_\ell}}{j_\ell!} \right) \,dt + O_{\Delta,\delta,d,\gamma}(n^{-d})$$ 
with the inner sum as above and we used the bound $\sigma = \Omega_{\Delta,\delta,\gamma}(n^{-1/2})$ to eliminate the logarithms in the error. Applying the identity \eqref{eq:Hermite-identity} and integrating completes the proof. 
\end{proof}

For later use, we isolate the $d = 2$ case of \cref{prop:edgeworth}. 

\begin{corollary}\label{cor:EE-application}
    In the setting of \cref{prop:edgeworth}, if we also have $|a| \leq L$ then $$\P(X = \mu + a) = \frac{e^{-a^2 / 2\sigma^2}}{\sqrt{2\pi}\sigma} + O_{\Delta,\delta,\gamma,L}(n^{-3/2})$$ 
\end{corollary}
\begin{proof}
    Apply \cref{prop:edgeworth} with $d = 2$.  Note that all sequences appearing in the sum have $j_5 = j_6 = \ldots = 0$, and the only choices for $(j_3,j_4)$ are $(1,0),(0,1)$ and $(2,0)$ implying $$
    \P(X = \mu + a) = \frac{e^{-a^2 / 2\sigma^2}}{\sqrt{2\pi}\sigma}\left(1 + H_3\left(\frac{a}{\sigma}\right)\beta_3 + H_4\left(\frac{a}{\sigma}\right)\beta_4 + H_6\left(\frac{a}{\sigma}\right)\frac{\beta_3^2}{2} \right) + O_{\Delta,\delta,\gamma}(n^{-2}\log^{6}n)\,. $$

    Recall that $\sigma^2 = \Omega_{\Delta,\delta,\gamma}(n)$ and apply \eqref{eq:beta-bound} along with the bound $|H_3(x)| \leq 3|x|$ for $|x| \leq 1$ to see 
    \[\frac{1}{\sigma}\left(1 + H_3\left(\frac{a}{\sigma}\right)\beta_3 + H_4\left(\frac{a}{\sigma}\right)\beta_4 + H_6\left(\frac{a}{\sigma}\right)\frac{\beta_3^2}{2} \right) = O_{\Delta,\delta,\gamma,L}(n^{-3/2})\,.\qedhere\] 
\end{proof}

\subsection{Properties of cumulants associated to the hard-core model via zero-free regions}

In this subsection, we prove \cref{lem:cumulant-bound} as well as the following crucial stability result. 

\begin{prop}\label{prop:cumulant-stability}
    Let $G$ be a graph on $n$ vertices of maximum degree $\Delta$ and $u \in V(G)$ and fix $\delta > 0$.  Let $\lambda \leq (1 - \delta)\lambda_c(\Delta)$. Let $X$ be the random variable giving the size of an independent set in the hardcore model on $G$ at activity $\lambda$ with $u$ conditioned on being in the independent set, and let $X'$ be the size conditioned on $u$ not being in the independent set.  Then for each $j \in \N$ we have \begin{equation*}
        | \kappa_j(X) - \kappa_{j}(X')| = O_{j,\Delta,\delta}(1)\,.
    \end{equation*}
\end{prop}

The engine used to prove these results is a zero-free region for the multivariate partition function for the hard-core model.  Given a function $\blambda: V(G) \to \C$ define $Z_G(\blambda)$ via 
\begin{equation*}
    Z_G(\blambda) = \sum_{I \in \mathcal{I}(G)} \prod_{v \in I} \blambda(v).
\end{equation*}
 Peters and Regts \cite{PR19} prove that $Z_G(\blambda)$ satisfies a certain zero-free property.

\begin{theorem}[Peters-Regts, \cite{PR19}] \label{thm:PR}
    Let $\Delta \geq 3$ and $\delta \in (0,1)$. There is an open set $U \subset \C$ with $[0,(1-\delta)\lambda_c(\Delta)] \subset U$ so that for any graph $G$ on $n$ vertices with maximum degree at most $\Delta$, we have $Z_G(\blambda) \neq 0$ provided $\blambda(v)  \in U$ for all $v \in G$.

    In particular, since $U$ is an open set, we may find a sufficiently small $\eps_{\mathrm{PR}}>0$, depending only on $\Delta, \delta$, so that $z \in U$ for all $|z| \leq \eps_{\mathrm{PR}}$ and $\lambda e^{w} \in U$ for all $\lambda \in [0,(1-\delta)\lambda_c(\Delta)]$ and $|w| \leq \eps_{\mathrm{PR}}$.
\end{theorem}

Our main use of the zero-free region provided by \cref{thm:PR} is to show analyticity of the logarithm of the partition function and control its derivatives.

\begin{fact}\label{fact:analytic-cgf}
   In the setting of \cref{thm:PR}, the function $t \mapsto \log \frac{Z_G(\lambda e^t)}{Z_G(\lambda)}$ is analytic for $|t| \leq \eps_{\mathrm{PR}}$.  Further, we have 
    \begin{equation}\label{eq:max-bound-cgf}
        \max_{|t| = \eps_{\mathrm{PR}}/2 } \left| \log \frac{Z_G(\lambda e^t)}{Z_G(\lambda)} \right| = O_{\Delta,\delta}(n)\,.
    \end{equation}
\end{fact}
\begin{proof}
    The function $t \mapsto \frac{Z_G(\lambda e^t)}{Z_G(\lambda)}$ is an analytic function since $Z_G$ is a polynomial.  By \cref{thm:PR}, we have that $Z_G(\lambda e^t)/Z_G(\lambda)\neq 0$ for $|t| \leq \eps_{\mathrm{PR}}$.  This shows analyticity of the function $t \mapsto \log \frac{Z_G(\lambda e^t)}{Z_G(\lambda)}$ for $|t| \leq \eps_{\mathrm{PR}}$.  To prove \eqref{eq:max-bound-cgf}, let $\{\zeta\}$ be the set of roots of $Z_G$ and note that for $|t| \leq \eps_{\mathrm{PR}}$ we have that $|\zeta - \lambda e^t|$ is uniformly bounded below by \cref{thm:PR}. 
    This shows 
    \[\max_{|t| = \eps_{\mathrm{PR}}/2 } \left| \log \frac{Z_G(\lambda e^t)}{Z_G(\lambda)} \right| = \max_{|t| = \eps_{\mathrm{PR}}/2 } \left|\log \left(\prod_{\zeta}\frac{\zeta - \lambda e^t}{\zeta - \lambda}\right)  \right| = O_{\Delta,\delta}(n)\,. \qedhere\]
\end{proof}

In a sample of the hard-core model from $G$ at activity $\lambda$, let $X$ be the random variable giving the size of the independent set sampled. The use of \cref{fact:analytic-cgf} is that it provides control of the cumulants of $X$, since we have $$\kappa_j(X) = \frac{d^j}{dt^j} \log \E e^{tX}|_{t=0} =  \frac{d^j}{dt^j}\log Z_G(\lambda e^t) \bigg|_{t = 0}\,. $$

First, we prove \cref{lem:cumulant-bound}.

\begin{proof}[Proof of \cref{lem:cumulant-bound}]
    The bound on $\kappa_j(X)$ follows immediately from \cref{fact:analytic-cgf} together with Cauchy's integral formula.  The second bound follows from noting that $$\left|\log \frac{Z_G(\lambda e^t)}{Z_G(\lambda)} - \sum_{j = 1}^d \kappa_j(X) \frac{t^j}{j!}\right| \leq \sum_{\ell > d} \frac{|\kappa_\ell(X)|}{\ell!}|t|^\ell$$ and applying the bound on $|\kappa_j(X)|$. 
\end{proof}

With \cref{lem:cumulant-bound} in place, we now shift to the proof of \cref{prop:cumulant-stability}. Let $G^{u}$ denote the graph obtained by removing $u$ and all of its neighbors from $G$, and let $G^{\ubar}$ denote the graph obtained by removing $u$ from $G$. With $X$ and $X'$ as in the statement of \cref{prop:cumulant-stability}, note that the moment generating function for $X$ is $\lambda e^t Z_{G^u}(\lambda e^t) / \lambda Z_{G^u}(\lambda)$ and the moment generating function $X'$ is $Z_{G^{\ubar}}(\lambda e^t)/Z_{G^{\ubar}}(\lambda).$   This means that for $j \in \N$ we have \begin{align*}
    \kappa_j(X) = \frac{d^{j}}{dt^j} \log \lambda e^t Z_{G^u}(\lambda e^t) \Big|_{t = 0} \\
    \kappa_j(X') = \frac{d^{j}}{dt^j} \log Z_{G^{\ubar}}(\lambda e^t) \Big|_{t = 0} 
\end{align*}
implying that \begin{equation}\label{eq:cumulant-difference}
    \kappa_j(X') - \kappa_j(X) =\frac{d^{j}}{dt^j} \log \left(\frac{\lambda e^t Z_{G^u}(\lambda e^t) }{Z_{G^{\ubar}}(\lambda e^t)}\right)  \Big|_{t = 0}\,.
\end{equation}

Let 
\begin{equation}
\label{eqn:occupancy-ratio}
R_u(\lambda,t) = \frac{\lambda e^t Z_{G^u}(\lambda e^t) }{Z_{G^{\ubar}}(\lambda e^t)}\,.
\end{equation}
This is exactly the so-called occupancy ratio, extended to complex activities. Analyticity of $R_u(\lambda,t)$ is immediate:

\begin{fact}\label{fact:analytic}
The function $t \mapsto R_u(\lambda,t)$ is analytic for $|t| \leq \eps_{\mathrm{PR}}$\,.
\end{fact}
\begin{proof}
    The function $R_u(\lambda,t)$ is a rational function in $e^t$ and by \cref{thm:PR} the denominator is zero-free for  $|t| \leq \eps_{\mathrm{PR}}$.
\end{proof}

By \cref{fact:analytic} and Cauchy's integral formula, \cref{prop:cumulant-stability} follows immediately, provided we can deduce a uniform upper bound on the magnitude of $R$. We will do this by applying \cref{thm:PR} to a well-chosen multivariate $\blambda$.

\begin{lemma}\label{lemma:bounded}
    There exist $C>0$, depending only on $\Delta,\delta$, such that for $|t| < \eps_{\mathrm{PR}}$ and $\lambda \in [0,(1-\delta)\lambda_c(\Delta)]$ we have $|R_u(\lambda,t)| \leq C$\,.
\end{lemma}
\begin{proof}
    For $\lambda \in [0,(1-\delta)\lambda_c(\Delta)]$ and $|t| \leq \eps_{\mathrm{PR}}$ define $z$ via $$R_u(\lambda, t) = z$$
    and rewrite \cref{eqn:occupancy-ratio} as \begin{align*}
        0 &= -\frac{\lambda e^t}{z} Z_{G^u}(\lambda e^t) + Z_{G^{\ubar}}(\lambda e^t) \,.
    \end{align*}
    Note that if we define the activity $\blambda$ via \begin{equation*}
        \blambda(v) = \begin{cases} 
        \lambda e^t &\text{ if } v \neq u \\
        -\frac{\lambda e^t}{z} & \text{ if }v = u
        \end{cases}
    \end{equation*}
    then we can write this equality as $Z_{G}(\blambda) = 0\,.$ By \cref{thm:PR}, since $\blambda(v) \in U$ for all $v\neq u$, it must be the case that $\blambda(u) \notin U$. In particular, we must have  $|\blambda(u)| = \left|\lambda e^{t}/z\right| \geq \eps_{PR}$, implying that 
    \[|z| \leq \lambda e^{|t|}/\eps_{PR} = O_{\Delta,\delta}(1)\,.\qedhere\]
\end{proof}

\begin{proof}[Proof of \cref{prop:cumulant-stability}]
    Apply equation \eqref{eq:cumulant-difference} to write $$ \left|\kappa_j(X') - \kappa_j(X) \right| =\left|\frac{d^{j}}{dt^j} \log R_u(\lambda,t)   \Big|_{t = 0}\right|\,.$$
    By \cref{fact:analytic} and \cref{lemma:bounded}, the function $t \mapsto R_u(\lambda,t)$ is analytic and bounded by $C = O_{\Delta,\delta}(1)$ for $|t| < \eps_{\mathrm{PR}}$.  Applying Cauchy's integral theorem completes the proof.  
\end{proof}

\subsection{Proof of \cref{th:SI-slice}} By the discussion at the start of this section, it suffices to show \eqref{eqn:prob-hitting}. This follows from at most two repeated applications of the following.

\begin{prop}\label{prop:hit-k}
    Let $G$ be a graph on $n$ vertices of maximum degree $\Delta$ and fix $\delta,\gamma,L  > 0$.  Then for  $\gamma \leq \lambda \leq (1 -\delta)\lambda_c(\Delta)$, $u \in V(G)$ and $k \in \N$ satisfying $|k - \E_\lambda |I| | \leq L$, we have \begin{equation}\label{eq:LCLT}
        \P_\lambda( |I| = k_1) = \Theta_{\Delta,\delta,\gamma,L}(n^{-1/2})
    \end{equation}
    %
    %
    %
    \begin{equation}\label{eq:L-TP}
    \left|\P_{\lambda}(|I| = k \,|\,u) - \P_\lambda(|I| = k)\right| = O_{\Delta,\delta,\gamma,L}(n^{-3/2}) 
    \end{equation}
    \begin{equation}
        \left|\P_{\lambda}(|I| = k \,|\,\ubar) - \P_\lambda(|I| = k)\right| = O_{\Delta,\delta,\gamma,L}(n^{-3/2}) 
    \end{equation}
\end{prop}

\begin{proof}[Proof of Proposition \ref{prop:hit-k}]
    The first line \eqref{eq:LCLT} follows immediately from \cref{cor:EE-application}. We prove \eqref{eq:L-TP}, noting that the remaining bound follows by a similar argument. By the law of total probability, it suffices to prove
    \begin{equation}
    \label{eqn:l-tp-1}
    \left|\P_{\lambda}(|I| = k \,|\,u) - \P_\lambda(|I| = k \,|\, \ubar)\right| = O_{\Delta,\delta,\gamma,L}(n^{-3/2})
    \end{equation}
    We bound

\begin{equation}\label{eq:mean-unconditional}
    \left|\E_\lambda[|I|] - \E_\lambda[|I| \,|\, u] \right| = \P_\lambda(\ubar)||\E_\lambda[|I| \,|\, \ubar] - \E_\lambda[|I| \,|\, u] | = O_{\Delta,\delta}(1),
    \end{equation}
    where the last bound is by \cref{prop:cumulant-stability} with $j = 1$. Let $|I|$ denote the size of the independent set sampled from the hard-core model on $G$ at activity $\lambda$, let $X$ denote the random variable $|I| \,|\, u$ and $X'$ the random variable $|I| \,|\, \ubar$.  Define $\mu = \E X$, $a = k - \mu$ and $\sigma^2 = \Var(X)$, and define $a',\mu',\sigma'$ analogously for $X'$.   \cref{prop:cumulant-stability} implies the bounds
    \begin{align*}
        |\mu - \mu'| &= |\kappa_1(X) - \kappa_1(X')| = O_{\Delta,\delta}(1)\\
        |\sigma - \sigma'| &= \frac{|\sigma^2 - \sigma'^{2}|}{|\sigma + \sigma'|} = \frac{|\kappa_2(X) - \kappa_2(X')}{|\sigma+\sigma'|} = O_{\Delta,\delta,\gamma}(1),
    \end{align*}
    where we have used that $\sigma, \sigma' = \Theta_{\delta,\Delta,\gamma}(\sqrt{n})$ \cite[Lemma~3.2]{jain2022approximate}.
    Therefore,
    $$\left|\frac{1}{\sigma} - \frac{1}{\sigma'}\right| = O_{\Delta,\delta,\gamma}(n^{-3/2}).$$ 
    Combining the first of these bounds with \eqref{eq:mean-unconditional} shows $|a| = O_{\Delta,\delta,L}(1)$ and $|a'| = O_{\Delta,\delta,L}(1)$. Therefore,
    \begin{align*}
        \left|\frac{e^{-a^2/2\sigma^2}}{\sigma} - \frac{e^{-a'^2/2\sigma'^2}}{\sigma'}\right| &\leq \left|\frac{1}{\sigma} - \frac{1}{\sigma'}\right| + \left|\frac{a^2}{2\sigma^3} - \frac{a^2}{2\sigma'^3}\right| + O_{\Delta,\delta, \gamma, L}(n^{-5/2})\\
        &= O_{\Delta,\delta,\gamma,L}(n^{-3/2}).
    \end{align*}
    \cref{cor:EE-application} now implies \eqref{eqn:l-tp-1}, which finishes the proof. 
\end{proof}

\subsection{Spectral gap of the down-up walk}
\label{sec:spectral-gap-independent-set}
We now combine \cref{th:SI-slice} with standard machinery from the theory of spectral independence to deduce that the down-up walk has optimal spectral gap $\Omega_{\delta,\Delta}(1/k)$ for all $k \leq (1-\delta)\alpha_c(\Delta)n$. 

\begin{theorem}
\label{thm:spectral-gap}
    Let $\Delta \geq 3$ and $\delta \in (0,1)$. For a graph $G = (V,E)$ on $n$ vertices of maximum degree at most $\Delta$ and an integer $1 \leq k\leq (1-\delta)\alpha_c(\Delta)n$, the down-up walk on independent sets of size $k$ satisfies a Poincar\'e inequality with constant $\Omega_{\delta,\Delta}(1/k)$. 
\end{theorem}

\begin{proof}
It suffices to prove the statement only for $k\geq n/(3\Delta)$, since for $k \leq n/(3\Delta)$, the statement follows from the work of Bubley and Dyer \cite{bubley1997path} along with a standard inequality relating contractive couplings and spectral gaps (see \cite[Theorem~13.1]{levin2017markov}): concretely, \cite{bubley1997path} shows that for any graph $H$ on $m$ vertices of maximum degree $\Delta$, for any $s\leq m/(3\Delta)$, and for $\mu_s^H$ the uniform distribution on $\mathcal{I}_s(H)$, the spectral gap of the down-up walk on $\mathcal{I}_s(H)$ is $\Omega_{\Delta}(1/k)$.

We will use \cref{th:SI-slice} to lift this to the entire range $k \leq (1-\delta)\alpha_c(\Delta)n$. To this end, let $G$ be as in the statement of the theorem and let $n/(3\Delta) \leq k \leq (1-\delta)\alpha_{c}(\Delta)n$. Note that for any independent set $J = \{v_1,\dots, v_{\ell}\} \subset V(G)$ of size $\ell \leq k - k\cdot c_\Delta$ (for a function $c_{\Delta} > 0$ of $\Delta$, to be specified later) we can view the conditional distribution $\mu_k(\cdot \,|\, I \supseteq J)$ as $\mu_{k-\ell}^{H}$, where $H$ is a graph with maximum degree at most $\Delta$ and with at least $n - \ell(\Delta+1)$ vertices. Since $k\leq n/(\Delta+1)$, we have 
\[c_{\Delta}\cdot \frac{k}{n} \leq \frac{k-\ell}{n-\ell(\Delta + 1)} \leq \frac{k}{n} \leq (1-\delta)\alpha_c(\Delta).\]
Therefore, by \cref{th:SI-slice}, $\mu_{k-\ell}^{H}$ is $O_{\Delta,\delta}(1)$-spectrally independent.

Let $\nu = \mu_{k}$ and consider the measures $\nu_1,\dots, \nu_k$ defined as follows: let $I \sim \nu$ and let $\{j_1,\dots, j_k\}$ be a random permutation of $[k]$, independent of $I$. Write $I = \{v_1,\dots, v_k\}$, where $v_1 < \dots < v_k$. For $i \in \{1,\dots,k\}$, define $\nu_i$ to be the law of {$I$ conditioned on $v_{j_1},\ldots,v_{j_i} \in I$}. 
Let $\nu_{t} = \nu_{\max(t,k)}$. By the $O_{\delta,\Delta}(1)$-spectral independence of $\mu^{H}_{k-\ell}$ established in the previous paragraph and \cite[Proposition~21, Fact~23]{chen2022localization}, it follows that for any $f:\mathcal{I}_k(G) \to \R$ and any $\ell \leq k-k\cdot c_\Delta$ 
\begin{equation}
\label{eqn:approx-conservation-variance}
    \frac{\E[\Var_{\nu_{k-\ell}}[f]]}{\Var_{\nu}[f]} = \Omega_{\delta,\Delta}(1).  
\end{equation}
Moreover, by the discussion above, it follows that almost surely, the measure $\nu_{k-\ell}$ can be represented as $\mu_{k-\ell}^{H}$ for a graph $H$ with maximum degree $\Delta$ and $m \geq n - \ell(\Delta+1)$ vertices. Let $c_\Delta = (1.01\Delta-3)/(6\Delta-3)$; note that for $\Delta \geq 3$, this is uniformly bounded away from $0$ and $1$. Let $\ell = k - \lfloor k\cdot c_{\Delta}\rfloor$. For $\Delta \geq 3$, since $k/n \leq \alpha_{c}(\Delta) \leq 0.74/\Delta$, it follows that
\begin{equation*}
    \frac{k-\ell}{n-\ell(\Delta+1)} \leq \frac{1}{3\Delta}. 
\end{equation*}
Therefore, it follows from the result of \cite{bubley1997path} mentioned in the first paragraph of the proof that with $c_\Delta$ and $\ell$ as above, almost surely, the spectral gap of the down-up walk on $\nu_{k-\ell}$ is $\Omega_{\Delta}(1/n)$. Combining this with \eqref{eqn:approx-conservation-variance} and the annealing technique of Chen and Eldan \cite[Theorem~46]{chen2022localization} shows that the spectral gap of the down-up walk on $\mu_k$ is
\[\Omega_{\delta,\Delta}(1)\cdot \Omega_{\Delta}(1/n) = \Omega_{\delta,\Delta}(1/n)\]
as desired. \qedhere
\end{proof}

\section{Initial reductions for the log-Sobolev inequality}
\label{sec:LSI}

The main goal of this section is to show that in order to prove \cref{thm:main-lsi}, it is enough to prove an LSI for $k$ small enough (\cref{lem:reduce to low k}).  Throughout this section, for independent sets $I,J$ we write $I \sim J$ if $I$ and $J$ can be obtained from each other by a single down-up step.

\subsection{Reduction to small $k$} We begin by showing (\cref{lem:reduce to low k}) that it is sufficient to prove a log-Sobolev inequality for the down-up walk for $\mu_k$ for $k$ sufficiently small. We begin with a simple lemma. 
\begin{lemma}\label{lem:marginal lower bound}
Let $G$ be a graph on $n$ vertices with maximum degree $\Delta$. For any $\gamma n \leq k \leq (1-\delta)\alpha_c(\Delta)n$ and any $u \in V(G)$, 
\[\min\{\mu_k[u], \mu_k[\ubar]\} = \Omega_{\delta,\Delta,\gamma}(1).\]
\end{lemma}
\begin{proof}
By \cref{lem:finding-activity} and \cref{prop:hit-k}, $\min\{\mu_k[u], \mu_k[\ubar]\}$ is within a constant factor of the corresponding quantity for the hard-core model at activity $\lambda = \Theta_{\gamma, \Delta,\delta}(1)$; the latter quantity is easily seen to be $\Omega_{\delta,\Delta,\gamma}(1)$ (see, e.g., \cite[Proposition 50]{AJKPV21b}).
\end{proof}

Recall the notation $\nu = \mu_k$, $\nu_1,\dots, \nu_k$ from the proof of \cref{thm:spectral-gap}. Combining the bounded-marginal property (\cref{lem:marginal lower bound}) with \cref{th:SI-slice}, it follows from \cite[Proposition~35, Theorem~42]{chen2022localization} (see also \cite{CLV20}) that for any $k = \Omega_{\Delta}(n)$, $\ell \leq k$ satisfying $k-\ell = \Omega_{\Delta}(k)$, and any non-negative function $f$, 
\begin{equation}
\label{eqn:entropy-conservation}
    \frac{\E[\Ent_{\nu_{k-\ell}}[f]]}{\Ent_{\nu}[f]} = \Omega_{\delta,\Delta}(1)
\end{equation}

We will show the following. 
\begin{prop}
\label{lem:reduce to low k}
There exist constants $C = O(1)$ and $C' = O_{\Delta}(1)$ for which the following holds: for all graphs $G$ on $n$ vertices with maximum degree $\Delta$ and for any $n/C' \leq k \leq n/(C\Delta^{8})$, the down-up walk on $\mu_k$ has log-Sobolev constant $\Omega_{\Delta}(1/n)$. 
\end{prop}

Given \cref{lem:reduce to low k} and \cref{eqn:entropy-conservation}, \cref{thm:main-lsi} follows from \cite[Proposition~48]{chen2022localization}: for any real-valued function $f$,
\begin{align*}
n\mathcal{E}_{P_\nu}(f,f)
&\geq n\E[\mathcal{E}_{P_{\nu_{k-\ell}}}(f,f)]\\
&\gtrsim_{\Delta} \E[\Ent_{\nu_{k-\ell}}[f^2]]\\
&\gtrsim_{\delta,\Delta}\Ent_{\nu}[f^2],
\end{align*}
where the first line follows from \cite[Proposition~48]{chen2022localization}, the second from \cref{lem:reduce to low k}, and the third from \cref{eqn:entropy-conservation}. 
\subsection{Reduction to graphs with small connected components} Next, we show that it is enough to prove \cref{lem:reduce to low k} in the case when the underlying graph $G$ has linearly many connected components, each of size $O_\Delta(\log{n})$. We will need the following crude bound on the log-Sobolev constant for the down-up walk.

\begin{lemma}
\label{lem:crude-LSI}
Let $G$ be a graph on $n$ vertices with maximum degree $\Delta$. For any $k \leq (1-\delta)\alpha_c(\Delta) n$, the down-up walk on $\mathcal{I}_k(G)$ satisfies a log-Sobolev inequality with constant $\Omega_{\delta,\Delta}(k^{-2}\log(n/k)^{-1})$.
\end{lemma}

\begin{proof}
    By \cref{th:SI-slice}, the down-up walk has spectral gap $\Omega_{\Delta,\delta}(1/k)$. Moreover, for any $I \in \mathcal{I}_k(G)$, $\mu_k(I) \geq \binom{n}{k}^{-1} \geq (en/k)^{-k}$. The assertion now follows immediately from a standard comparison between the spectral gap and log-Sobolev constant (\cite[Corollary~A.4]{diaconis1996logarithmic}). 
\end{proof}

We will also need the following definition and elementary combinatorial lemma.

\begin{definition}
    \label{def:good-graph}
Let $G$ be a graph on $n$ vertices and $\Delta \geq 3$. We say that $G$ is $\Delta$-good if the maximum degree of $G$ is at most $\Delta$, the maximum connected component of $G$ has size at most $1000\Delta\log{n}$, and moreover, the number of connected components is at least $\Omega_{\Delta}(n)$. %
\end{definition}

\begin{lemma}
\label{lem:connected-components}
    Let $G$ be a graph on $n$ vertices with maximum degree $\Delta$. Let $k \leq n/\Delta^{8}$ and fix a subset $S$ of size $k$ such that all connected components of the induced graph $G|_S$ have size at most $2$.  Let $W$ be a random subset of size $\ell = \lfloor n/\Delta^6\rfloor$ containing $S$. Then, with probability (over $W$) at least $1-n^{-50}$, the graph $G|_{W}$ is $\Delta$-good. 
\end{lemma}
\begin{proof}
Let $C$ be a connected component of $G|_W$. Then for any distinct $u,v \in C$, there exists a shortest path $u = u_0, u_1,\dots, u_m = v$, where $u_0,\dots, u_m \in W$ and $u_i$ is connected to $u_{i+1}$ in $G$. Since $S$ has connected components of size at most $2$, it follows that $C \setminus S$ %
is connected in $G^{(3)}$ -- the graph on the same vertex set as $G$ and with $u,v$ connected if there is a path of length at most $3$ in $G$ from $u$ to $v$. From this argument, it also follows that the the neighborhood in $G$ of $C\setminus S$ contains all but at most $|C\cap S|/2$ vertices of $C$. Since $G$ has maximum degree $\Delta$, it follows that 
\[|C \setminus S| \geq |C|/(10\Delta).\]
It follows that the probability that the maximum connected component of $G|_W$
 is at least $1000\Delta\log{n}$ is bounded above by the probability that the maximum connected component of $G^{(3)}|_{W\setminus S}$ is at least $100\log{n}$. 
Since $G^{(3)}$ is a graph on $n$ vertices of maximum degree at most $\Delta^{3}$ and $W\setminus S$ is a uniformly random subset of $[n]\setminus S$ of size less than $\ell$, this follows immediately from \cite[Lemma~4.3]{CLV20} and the union bound. 

Similarly, for a given vertex $v$, the probability that the connected component of $v$ in $G|_W$ has size at least $h$ is bounded by the probability that the connected component of $v$ in $G^{(3)}|_{W\setminus S}$ has size at least $h/10\Delta$, which is at most $e^{-h/1000\Delta}$ by \cite{CLV20}.  This shows that the expected number of connected components in $G|_W$ is linear in $n$. 
 We then reveal each vertex in $W$ one at a time and consider the number of connected components as a martingale; we see that each revealed vertex can change the number of connected components by at most $O_\Delta(1)$, and so Azuma's inequality shows that the number of components is $\Omega_{\Delta}(n)$ with probability $1 - \exp(-\Omega_\Delta(n)) \le 1-n^{-100}$.%
\end{proof}

In \cref{sec:lee-yau}, we will show the following. 

\begin{prop}
\label{prop:lee-yau-regular}
    There exist constants $C = O(1)$ and $C' = O_{\Delta}(1)$ for which the following holds: for all graphs $G$ on $n$ vertices which are $\Delta$-good and for all $n/C' \leq k \leq n/(C\Delta^{8})$, the down-up walk on $\mathcal{I}_k(G)$ satisfies a log-Sobolev inequality with constant $\Omega_{\Delta}(1/n)$. %
\end{prop}

Given these preliminaries, we can prove \cref{lem:reduce to low k}. 

\begin{proof}[Proof of \cref{lem:reduce to low k}]
    Let $\bar{\mu}_k$ denote the complement distribution of $\mu_k$ on $\binom{[n]}{n-k}$, i.e.~for any $J \in \binom{[n]}{n-k}$, $\bar{\mu}_k(J) := \mu_k(\bar{J})$. As in the proof of \cref{thm:spectral-gap}, we construct a process $\bar{\nu} = \bar{\mu}_k, \bar{\nu}_1,\dots, \bar{\nu}_{n-k}$: {in particular, we let $w_1 < \ldots<w_{n - k}$ be a sample of $J$ and set $\bar{\nu}_{i}$ to be the law of $J$ conditioned on $w_{j_1},\ldots,w_{j_i} \in J$ where $\{j_1,\ldots,j_{n-k}\}$ is a random permutation of $[n-k]$}. Let $\ell = \lfloor n/\Delta^6\rfloor $, set $t = n - \ell$ and note that for $r \leq t$, $\bar{\nu}_{r}$ almost surely corresponds to taking the complement of an independent set of size $k$ chosen uniformly at random from a graph $H$ of maximum degree $\Delta$ with $n - (n- r) \geq \ell$ vertices. 
   
    Therefore, by \cref{lem:marginal lower bound}, the measures $\bar{\nu},\bar{\nu}_1,\dots, \bar{\nu}_{t}$ almost surely have uniformly lower bounded marginals. Hence, by \cite[Proposition~35, Theorem~42]{chen2022localization}, for any non-negative function $g$ on $\binom{[n]}{n-k}$,
    \begin{align}\label{eq:nubar-nubar-t}
        \Ent_{\bar{\nu}}[g] = O_{\delta,\Delta}(1)\E[\Ent_{\bar{\nu}_{t}}[g]]
    \end{align}
    Let $\omega_{t}$ denote the distribution on $\binom{[n]}{t}$ induced by $\bar{\mu}_k$, i.e. for $T \in \binom{[n]}{t}$,
    \[\omega_{t}(T) = \binom{n-k}{t}^{-1}\cdot \bar{\mu}_k\left[S \in \binom{[n]}{n-k}, T\subseteq S\right].\]
    Then, for any function $f:\mathcal{I}_k(G) \to \R_{\geq0}$, denoting $\bar{f}$ by $g$, we have
    \begin{align*}
    \Ent_{\mu_k}[f] 
    &= \Ent_{\bar{\mu}_k}[g]\\
    &\lesssim_{\delta,\Delta} \E[\Ent_{\bar{\nu}_t}[g]]\\
    &= \sum_{T \in \binom{[n]}{t}}\omega_{t}(T)\Ent_{\bar{\nu}(\cdot \mid T)}[g]\\
    &\lesssim_{\delta,\Delta}\sum_{T \in \binom{[n]}{t}}\omega_{t}(T)\rho_T^{-1}\mathcal{E}_T(\sqrt{g},\sqrt{g}).
    \end{align*}
    Here, we set $H_T$ to be the graph maximum degree $\Delta$ graph $G \setminus \{w_{i_1},\ldots,w_{i_t}\}$ on $\ell$ vertices; identify $\bar\nu(\cdot \mid T)$ with the complement of an independent set chosen from $\mu_k^{H_T}$;  $\mathcal{E}_T$ is the Dirichlet form of the up-down walk on $\overline{\mathcal{I}_k(H_T)}:= \{\bar{I} : I \in \mathcal{I}_k(H_T)\}$; and $\rho_{T}$ is the log-Sobolev constant of this walk. 

    Let $\mu_T$ denote the uniform distribution on $\mathcal{I}_k(H_T)$ and $P_T$ denote the transition matrix of the down-up walk on $\mathcal{I}_k(H_T)$. Then, we have
    \begin{align*}
        \omega_T(T)\rho_T^{-1}\mathcal{E}_T(\sqrt{g},\sqrt{g})
        &= \omega_T(T)\rho_T^{-1}\sum_{I\sim J \in \mathcal{I}_k(H_T)}\mu_T(I)P_T(I,J)(\sqrt{f(I)} - \sqrt{f(J)})^2\\
        &= \omega_T(T)\rho_T^{-1}\sum_{I\sim J \in \mathcal{I}_k(H_T)}\frac{\mu_k(I)}{\mu_k(I':I'\subseteq H_T)}P_T(I,J)(\sqrt{f(I)} - \sqrt{f(J)})^2\\
        &= \binom{n-k}{t}^{-1}\rho_T^{-1}\sum_{I\sim J \in \mathcal{I}_k(H_T)}\mu_k(I)P_T(I,J)(\sqrt{f(I)} - \sqrt{f(J)})^2.
    \end{align*}
Therefore,    
    \begin{align*}
       \binom{n-k}{t}\sum_{T \in \binom{n}{t}}\omega_t(T)\rho_T^{-1}\mathcal{E}_T(\sqrt{g},\sqrt{g})
        &= \sum_{T \in \binom{n}{t}}\rho_T^{-1}\sum_{I\sim J \in \mathcal{I}_k(H_T)}\mu_k(I)P_T(I,J)(\sqrt{f(I)} - \sqrt{f(J)})^2\\
        &= \sum_{I\sim J \in \mathcal{I}_k(G)}\mu_k(I)(\sqrt{f(I)}-\sqrt{f(J)})^{2}\alpha_t(I,J),
    \end{align*}
with
\begin{align*}
\alpha_t(I,J)
&= \sum_{T: I\sim J \in \mathcal{I}_k(H_T)}\rho_T^{-1}P_T(I,J) \lesssim_{\Delta}P(I,J)\sum_{T: I\sim J \in \mathcal{I}_k(H_T)}\rho_T^{-1},
\end{align*}
where $P$ is the transition matrix of the down-up walk on $G$. In the last inequality, we have used that the probability of transitioning from $I$ to $J$ in $H_T$ is at most a constant factor (depending on $\Delta$) more than the probability of transitioning from $I$ to $J$ in $G$ since $\ell - (\Delta+1)k = \Theta_{\Delta}(n)$. 

Finally, for $I\sim J \in \mathcal{I}_k(G)$, let $\mathcal{T}_g(I,J)$ denote the set of $T \in \binom{n}{t}$ such that $I\sim J \in \mathcal{I}_k(H_T)$ and $H_T$ is $\Delta$-good; let $\mathcal{T}_b(I,J) = \{T: I,J \in \mathcal{I}_k(H_T)\}\setminus \mathcal{T}_g(I,J)$. Then,
\begin{align*}
    \sum_{T: I\sim J \in \mathcal{I}_k(H_T)}\rho_T^{-1}
    &= \sum_{T \in \mathcal{T}_g(I,J)}\rho_T^{-1} + \sum_{T \in \mathcal{T}_b(I,J)}\rho_T^{-1}\\
    &\lesssim_{\delta,\Delta} |\mathcal{T}_g(I,J)|n + |\mathcal{T}_b(I,J)|n^{2}\\
    &\lesssim_{\delta,\Delta} \binom{n-(k+1)}{\ell-(k+1)}\cdot \left(n + \binom{n-(k+1)}{\ell-(k+1)}^{-1}\cdot |\mathcal{T}_b(I,J)|n^{2}\right)\\
    &\lesssim_{\delta,\Delta} \binom{n-(k+1)}{\ell-(k+1)}\cdot n,
\end{align*}
where the second line uses \cref{prop:lee-yau-regular} and \cref{lem:crude-LSI}, and the last line uses \cref{lem:connected-components}. 

Putting everything together, we have
\begin{align*}
    \Ent_{\mu_k}[f]
    &\lesssim_{\delta,\Delta} \sum_{I\sim J \in \mathcal{I}_k(G)}\mu_k(I)P(I,J)(\sqrt{f(I)}-\sqrt{f(J)})^2\cdot n\cdot \frac{\binom{n-k-1}{t}}{\binom{n-k}{t}}\\
    & \lesssim_{\delta,\Delta} n\cdot \sum_{I\sim J \in \mathcal{I}_k(G)}\mu_k(I)P(I,J)(\sqrt{f(I)}-\sqrt{f(J)})^2\\
    & \lesssim_{\delta,\Delta}n\cdot \mathcal{E}_P(\sqrt{f},\sqrt{f}),
\end{align*}
where the second line uses that $n-k = \Theta(t)$. \qedhere
\end{proof}

\section{Log-Sobolev inequality: Proof of \cref{thm:main-lsi}} 
\label{sec:lee-yau}

It remains to prove \cref{prop:lee-yau-regular}. We will find it more convenient to consider a slightly different Markov chain, where we only perform the down-up walk between different connected components. To this end, we record the following definition. 

\begin{definition}
The modified down-up walk on $\mathcal{I}_k(G)$ is defined as follows: given the current state $I_t$ at time $t$, we choose independent and uniform vertices $u,v \in V(G)$. If $v$ and $u$ are in different connected components of $G$ and $I' = (I\setminus u) \cup v \in \mathcal{I}_k(G)$, then $I_{t+1} = I'$; else $I_{t+1} = I_t$. 
\end{definition}

Note that the modified down-up walk is reversible with respect to $\mu_k(G)$. \cref{prop:lee-yau-regular} is a consequence of the following.

\begin{theorem}\label{thm:lee-yau}
There exists a constant $C_{\ref{thm:lee-yau}} = O(1)$ for which the following holds: for all graphs $G$ on $n$ vertices which are $\Delta$-good and for all $n/(C_{\ref{thm:lee-yau}}\Delta^{1000}) \leq k \leq n/(C_{\ref{thm:lee-yau}}\Delta^{8})$, the modified down-up walk on $\mathcal{I}_k(G)$ satisfies a log-Sobolev inequality with constant $\Omega_{\Delta}(1/n)$.  
\end{theorem}

Before proving \cref{thm:lee-yau}, let us quickly show that it implies \cref{prop:lee-yau-regular}.

\begin{proof}[Proof of \cref{prop:lee-yau-regular} given \cref{thm:lee-yau}] Let $P$ denote the transition matrix of the down-up walk and let $Q$ denote the transition matrix of the modified down-up walk. It suffices to show that for all non-negative functions $f$ on $\mathcal{I}_k(G)$,
\[\mathcal{E}_Q(f,f) \lesssim_{\Delta} \mathcal{E}_P(f,f).\]
For this, we simply note that for any $I\sim J \in \mathcal{I}_k(G)$, with $I \neq J$
\[Q(I,J) \lesssim_{\Delta} P(I,J);\]
indeed, if $J = (I\setminus u) \cup v$ with $u,v$ in the same connected component, then $Q(I,J) = 0$, whereas if $u,v$ are in different connected components, then $Q(I,J) = \Theta_{\Delta}(1/n^2) =  \Theta_{\Delta}(P(I,J)).$ \qedhere

\end{proof}

We prove \cref{thm:lee-yau} in three stages: first we prove an optimal log-Sobolev inequality for the chain induced by the modified down-up walk on a single connected component of the graph, where we recall that such components are all fairly small (\cref{thm:compare between hardcore and induced distribution on one component}); in the spirit of Lee-Yau \cite{lee1998logarithmic}, we show a recurrence for the inverse LSI constant of a graph in terms of the graph with a single component removed (\cref{prop:recurrence}); a key ingredient in establishing this recurrence is an optimal Markov chain decomposition-type result in the presence of contractive couplings, based on Stein's method for Markov chains (\cref{prop:opt-decomp}); finally, we analyze this recurrence via a martingale argument. 

\subsection{The induced chain}

Given a graph $\mathcal{G}$ with different connected components $\{G\}$, let $I$ be a uniformly random independent set of size $k$ on $\mathcal{G}$. We denote $I_G$ the projection of $I$ on the component $G$, and $I_{-G}$ the projection on $\mathcal{G}-G$. We denote $P_{u,v}$ the operator that swaps the status of $u$ and $v$ in the independent set if $u$ and $v$ are in different connected components and if the swapping yields a valid independent set. Concretely, by swapping, we mean that if $u$ and $v$ are both occupied or both unoccupied, we do nothing; if $u$ is unoccupied and $v$ is occupied, then $u$ becomes occupied and $v$ becomes unoccupied (and vice versa). We define the induced distribution on $G$ as the distribution of $I_G$ on $2^G$ and define the following Markov chain on independent sets of $G$:
\begin{equation}\label{eq:PG-def}
P_G(I_G,J_G) = \frac{1}{|G|} \sum_{u\in G} \mathbb{P}_{I_{-G}|I_G,v \sim \mathcal{G}} [J_G=(P_{u,v}(I))_G].
\end{equation}
The importance of the above Markov chain is that its stationary distribution is the distribution of $I_G$ where $I$ is sampled uniformly at random from $\mathcal{I}_k(\mathcal{G})$. Indeed, given $I_G,J_G$ independent sets on $G$ differing in exactly one vertex $v$ with $v\in I_G$ and $v\notin J_G$, then $I_{-G}$ (respectively $J_{-G}$) is uniformly distributed among independent sets of size $k-|I_G|$ (respectively $k-|J_G|$) in $\mathcal{G}-G$. The above chain satisfies 
\[
\frac{{P}_G(I_G,J_G)}{{P}_G(J_G,I_G)} = \frac{\E[M(K)]}{(k-|J_G|)}.
\]
Here, $K$ is uniformly random in $\mathcal{I}_{k-|I_G|}(\mathcal{G}-G)$, and $M(K)$ is the number of vertices not in $K$ that form an independent set when included with $K$. On the other hand, in the bipartite graph with two sides being $\mathcal{I}_{k-|J_G|}(\mathcal{G}-G)$ and $\mathcal{I}_{k-|I_G|}(\mathcal{G}-G)$ and two sets adjacent if and only if they differ in exactly one vertex, the average degree of the first side is $(k-|J_G|)$ and of the second side is $\E[M(K)]$. Hence,
\begin{equation}\label{eq:M(K)}
\frac{\E[M(K)]}{k-|J_G|} = \frac{|\mathcal{I}_{k-|J_G|}(\mathcal{G}-G)|}{|\mathcal{I}_{k-|I_G|}(\mathcal{G}-G)|}.
\end{equation}

We will need control over the log-Sobolev constant of the above chain, which we denote by $\rho_G^{\rm{induced}}$. This is the content of the next lemma.   

\begin{lemma}\label{thm:compare between hardcore and induced distribution on one component} 
    Let  $ \mathcal{G}$ be a graph of maximum degree at most $\Delta$ and let $G$ be a connected component of $\mathcal{G}$. Assume that $|G| \le |\mathcal{G}|^{1/4}$. Suppose that $\alpha = k/\abs{\mathcal{G}} \in [\Delta^{-1000}, \Delta^{-2}]$. Then, $\rho_G^{\textrm{induced}} = \Omega_{\Delta}(|G|^{-1})$. 
\end{lemma}

\begin{proof}
    We will show that the induced chain has stationary distribution and transition probabilities within $O_{\Delta}(1)$ of the hardcore model on $G$ with activity $\lambda = \alpha$. By the Local Central Limit Theorem for the hardcore model in \cite{jain2022approximate} (or \cref{prop:edgeworth}), for integers $a,b \le |G| \le |\mathcal{G}|^{1/4}$ and $\alpha \in [\Delta^{-1000},\Delta^{-2}]$, the ratio 
    \begin{equation}\label{eq:slice-ratio}
        \frac{\lambda^{-a}|\mathcal{I}_{k-a}(\mathcal{G}-G)|}{\lambda^{-b}|\mathcal{I}_{k-b}(\mathcal{G}-G)|} = O_{\Delta}(1).%
    \end{equation}
    This implies that for all $I_G$ 
    \[
        \frac{\mu_k(I_G)}{\mu_\lambda(I_G)} = O_{\Delta}(1). 
    \]
    Moreover, the transition probability of the induced chain from $I_G$ to $J_G$, where $v\in I_G$ and $v\notin J_G$ and $I_G,J_G$ differ only at $v$, satisfies
    \[
        \frac{P_G(I_G,J_G)}{P_G(I_G,I_G)} = \frac{\E[M(K)]}{|\mathcal{G}-G|-\E[M(K)]},
    \]
    with notation as above. Similarly, 
    \[
        \frac{P_G(J_G,I_G)}{P_G(J_G,J_G)} = \frac{k-|J_G|}{|\mathcal{G}-G|-k+|J_G|}.
    \]
    From (\ref{eq:M(K)}) and (\ref{eq:slice-ratio}), we have that 
    \[
        \frac{P_G(I_G,J_G)}{P_G(I_G,I_G)} = \Theta_{\Delta}(k/n),
    \]
    and 
    \[
        \frac{P_G(J_G,I_G)}{P_G(J_G,J_G)} = \Theta_{\Delta}(k/n).
    \]
    Hence, the transition probabilities of the induced chain are within $O_{\Delta}(1)$ of those of the hardcore model of $G$ with activity $\lambda$. 

 Under the assumption $\lambda \in [\Delta^{-1000},\Delta^{-2}]$ and $G$ has maximum degree at most $\Delta$, we have that the hardcore model $\mu_\lambda$ on $G$ has LSI with constant $\Omega(1/|G|)$ (see, e.g.,\cite[Fact~3.5]{CLV20}). Thus, by standard comparison results (\cite[Theorem~2.14]{montenegrotetali}), we also have the $\rho^{\textrm{induced}}_{G}  \geq \Omega_\Delta(\abs{G}^{-1})$. 
\end{proof}

\subsection{Setting up the reduction}

We denote a \emph{configuration} $\cC$ to be a graph $\cG$ together with an integer $k$.  We denote the occupancy ratio $\alpha = k/|\cG|$.  

Given a configuration $\mathcal{C}$, we denote by $\rho_{\mathcal{C}}$ the inverse LSI constant associated to $\mathcal{C}$ multiplied by $1/|\mathcal{G}|$, and $\mu_{\mathcal{C}}$ the distribution of independent sets of size $k$ in $\mathcal{G}$.  In particular, for non-negative functions $f$ on independent sets $\cG$ of size $k$, we have \begin{equation}\label{eq:rho-cc-def}
    \Ent(f) \leq |\cG| \rho_{\cC} \cE_{\cC}(\sqrt{f},\sqrt{f})
\end{equation} Here we choose this normalization for $\rho_{\cC}$ because our goal for \cref{thm:lee-yau} is to prove $\rho_{\cC} = O_{\Delta}(1)$.  
For an independent set $I \in \cI_k(\cG)$ and set of vertices $G \subset \mathcal{G}$, %
we denote $I_G = I \cap G$ %
and $I_{-G} = I \setminus G$. Given a graph $\cG$ and a connected component $G$ of $\cG$, we note that the distribution of $I_{-G}$ is associated to a configuration $\cC'$ on the graph $\cG - G$ with $k' = k - |I_G|$.  

The main goal of this subsection is to prove the following recurrence for $\rho_{\cC}.$

\begin{prop}\label{prop:recurrence}
    There exists an absolute constant $C_{\ref{prop:recurrence}} = O(1)$ for which the following holds. Let $\cG$ be a graph on $n$ vertices of maximum degree $\Delta$.  Suppose $\cG$ has $m$ connected components, each of which is $\Delta$-good. Then for any integer $k$ with $1/(C_{\ref{prop:recurrence}}\Delta^{1000}) \leq k/n \leq 1/(C_{\ref{prop:recurrence}}\Delta^8)$, we have
$$ \rho_{\cC} \leq \frac{m-2}{m} \max_{I\sim J} \frac{1}{m-2}\left(\sum_{G:I_G=J_G} \left(\frac{|\mathcal{G}|}{|\mathcal{G}|-|G|} \rho_{\mathcal{C}'}\right)\right) + O_{\Delta}(m^{-1}) \,.$$
\end{prop}

Our starting point is the following basic fact about entropy: for any non-negative function $f$ on independent sets and any set of vertices $G \subset \mathcal{G}$, the following decomposition holds:
\begin{equation}\label{eq:Ent-cond-G}
\Ent(f) = \E_{I_G}[\Ent f|_{I_G}] + \Ent(\mathbb{E}[f|_{I_{G}}]) = \E_{I_G}[\Ent f|_{I_G}] + \Ent(f_G)
\end{equation}
Here we denote $f|_{I_G}$ the restriction of $f$ over independent sets on $\mathcal{G}-G$ conditional on $I_G$ and let $f_G = \E[f|I_G]$ denote the conditional expectation (which is a function on independent sets of $G$).  

A key ingredient in our proof is the following comparison, whose (non-trivial) proof is deferred to \cref{sec:opt-decomp}.
\begin{prop}\label{prop:opt-decomp}
Under the assumptions of \cref{prop:recurrence}, let  $G$ be a connected component of $\cG$, $I_G$ an independent set of $G$, and $u \in I_G$.  Then for any non-negative $f$ we have
    \begin{align*}
      \left(\sqrt{f_G(I_G)}-\sqrt{ f_G(I_G \setminus u)}\right)^2
    &\lesssim_{\Delta}  \E_{v}\E_{I_{-G}|I_G}\left[\left(\sqrt{f(I)}-\sqrt{f(P_{u,v}(I))}\right)^2\right],
    \end{align*}
\end{prop}

\cref{prop:opt-decomp} readily implies the following. 

\begin{lemma}\label{lem:Ent-term-1}
    Under the assumptions of \cref{prop:recurrence} we have $$\sum_G \Ent(f_G) \lesssim_{\Delta} n \cE_{\cC}(\sqrt{f},\sqrt{f})$$
    where the sum is over the connected components of $\cG$.
\end{lemma}
\begin{proof}
    Observe that 
    \[
    \mathcal{E}_{\cC}(\sqrt{f},\sqrt{f}) = \E_{I,u,v}\left[\left(\sqrt{f(I)}-\sqrt{f(P_{u,v}(I))}\right)^2\right]\,,
    \]
    where the operator $P_{u,v}$ is as above, and the expectation is taken according to $I \sim \mu_k$ and $u,v$ chosen uniformly and independently from $\mathcal{G}$. 
    We now let $\{G\}$ denote the connected components of $\mathcal{G}$ and will break apart the above expectation by which component $u$ falls in.  For this, it will be convenient to choose a component $G$ randomly with probability $|G|/|\cG|$. We may then rewrite
    \begin{align}
    \label{eqn:rewrite-main-dirichlet}
    \E_{I,u,v}\left[\left(\sqrt{f(I)}-\sqrt{f(P_{u,v}(I))}\right)^2\right] &= \frac{1}{n} \sum_{G} \sum_{u\in G} \E_{I_G}\E_{I_{-G}|I_G,v} \left[\left(\sqrt{f(I)}-\sqrt{f(P_{u,v}(I))}\right)^2\right] \nonumber\\
    &= \E_{G}\E_{u \in G}\E_{I_G} \E_{v} \E_{I_{-G}|I_G} \left[\left(\sqrt{f(I)}-\sqrt{f(P_{u,v}(I))}\right)^2\right].
    \end{align}

     On the other hand, by \cref{thm:compare between hardcore and induced distribution on one component}, we have 
    \begin{equation}\label{eq:Ent-induced}
    \Ent(f_G) \lesssim_{\Delta} |G|\mathcal{E}_G(\sqrt{f_G}, \sqrt{f_G})
    \end{equation}
    where $\mathcal{E}_G$ is the Dirichlet form of the induced chain on $G$. %
    We can bound the Dirichlet form of the induced chain by 
    \begin{align}
    \label{eqn:decomposition}
    \mathcal{E}_{G}(\sqrt{f_G},\sqrt{f_G})
    &= \E_{I_G}\sum_{J_G}P_G(I_G, J_G)\left(\sqrt{f_G(I_G)} - \sqrt{f_G(J_G)}\right)^2 \nonumber \\
    &\lesssim \E_{I_G} \E_{u \in G} \left(\sqrt{f_G(I_G)}-\sqrt{f_G(I_G \setminus u)}\right)^2 \nonumber \\
    &\lesssim_{\Delta} \E_{I_G} \E_{u\in G,v}\E_{I_{-G}|I_G}\left[\left(\sqrt{f(I)}-\sqrt{f(P_{u,v}(I))}\right)^2\right],
    \end{align}
    where the second line uses the reversibility of $P_G$ and $P_G(I_G, J_G) \lesssim |G|^{-1}$ and the last inequality uses \cref{prop:opt-decomp}. 

    Combining this with \eqref{eq:Ent-induced}, we have 
    \[\Ent(f_G) \lesssim_{\Delta} |G| \E_{I_G} \E_{u\in G,v}\E_{I_{-G}|I_G}\left[\left(\sqrt{f(I)}-\sqrt{f(P_{u,v}(I))}\right)^2\right]\]
    Finally, summing over the connected components $G$ of $\cG$ and using \cref{eqn:rewrite-main-dirichlet}, we get
    \begin{align*}
    \sum_G \Ent(f_G) &\lesssim_{\Delta} n  \E_G \E_{I_G} \E_{u\in G,v}\E_{I_{-G}|I_G}\left[\left(\sqrt{f(I)}-\sqrt{f(P_{u,v}(I))}\right)^2\right]\\
    &= n   \cE_{\cC}(\sqrt{f},\sqrt{f})\,. \qedhere
    \end{align*}
  
\end{proof}

Below, for two independent sets $I,J$ of $\mc{G}$, we say that $I\sim J$ if there are two vertices $u$ and $v$ in different connected components such that $J = (I\setminus u) \cup v$. In other words, $J = P_{u,v}(I)$ for a pair of vertices $u,v$ (in different connected components of $\mc{G}$). 

\begin{lemma}\label{lem:Ent-term-2}
    Under the assumptions of  \cref{prop:recurrence}
    $$\sum_{G} \E_{I_G} [\Ent (f|_{I_G})] \leq  n^2 \cE_{\cC}(\sqrt{f},\sqrt{f}) \max_{I\sim J} \sum_{G : I_G = J_G} \frac{\rho_{\cC'}}{|\cG| - |G|}  $$
    where the sum is over the connected components of $\cG$.
\end{lemma}
\begin{proof}
    Note that conditional on $I_{G}$, the distribution of $I$ on the remaining components is given by a configuration $\mathcal{C}'$ on $\mathcal{G}-G$ and $k'=k - |I_G|$. For any $I_G$, by the definition of $\rho_{\mathcal{C}'}$ \eqref{eq:rho-cc-def} we have
    \begin{equation}\label{eq:Ent-induction}
    \Ent (f|_{I_G}) \le ({|\mathcal{G}|-|G|})\rho_{\mathcal{C}'} \mathcal{E}_{\mathcal{C}'}(\sqrt{f|_{I_G}}, \sqrt{f|_{I_G}}). 
    \end{equation}
 Note that for each  fixed $I_G$ we have \begin{equation*}
        \cE_{\cC'}(\sqrt{f|_{I_G}},\sqrt{f|_{I_G}}) =  \sum_{J\sim L : J_G = L_G = I_G}  \mu(J_{-G}\,|\,I_G) \cdot \frac{1}{(|\cG| - |G|)^2} \left(\sqrt{f(J)} - \sqrt{f(L)} \right)^2\,.
    \end{equation*}

    Summing over $G$ provides \begin{align*} \sum_{G} ({|\mathcal{G}|-|G|})\mathbb{E}_{I_{G}}& \rho_{\mathcal{C}'} \mathcal{E}_{\mathcal{C'}}(\sqrt{f|_{I_G}},\sqrt{f|_{I_G}}) \\
    &= 
        \sum_{G} \frac{1}{|\cG| - |G|}\sum_{I_G}\mu({I_G}) \mu(J_{-G}\,|\,I_G) \rho_{\cC'} \sum_{J\sim L : J_G = L_G = I_G} \left(\sqrt{f(J)} - \sqrt{f(L)} \right)^2 \\
    &= \sum_{I \sim J} \sum_{G : I_G = J_G} \frac{\rho_{\cC'}}{|\cG| - |G|} \mu(I) \left(\sqrt{f(I)} - \sqrt{f(J)} \right)^2 \\
    &\leq |\cG|^2 \cE_{\cC}(\sqrt{f},\sqrt{f}) \max_{I\sim J} \sum_{G : I_G = J_G} \frac{\rho_{\cC'}}{|\cG| - |G|}\,. \qedhere
    \end{align*}
\end{proof}

\begin{proof}[Proof of \cref{prop:recurrence}]
    We now apply \eqref{eq:Ent-cond-G} for every connected component $G$ in $\cG$, using \cref{lem:Ent-term-1} for the first terms and \cref{lem:Ent-term-2} for the second terms to see  
    \begin{equation}\label{eq:Ent-sum-components}
    m\Ent(f) \leq \mathcal{E}_{\cC}(\sqrt{f},\sqrt{f}) \left( O_{\Delta}(n)  + n\cdot \max_{I\sim J} \sum_{G : I_G = J_G} \rho_{\cC'}\frac{|\cG|}{|\cG| - |G|} \right) 
    \end{equation}
    where we recall $m$ is the number of connected components of $\cG$. Recalling the definition of $\rho_{\cC}$, we get 
    \begin{align*}
    \rho_{\mathcal{C}}&\le\frac{1}{m} \max_{I\sim J} \sum_{G:I_G=J_G}\rho_{\mathcal{C}'} \frac{|\mathcal{G}|}{|\mathcal{G}|-|G|}+ O_{\Delta}(m^{-1})  \\
    &= \frac{m-2}{m} \max_{I\sim J} \frac{1}{m-2}\left(\sum_{G:I_G=J_G} \left(\frac{|\mathcal{G}|}{|\mathcal{G}|-|G|} \rho_{\mathcal{C}'}\right)\right) + O_{\Delta}(m^{-1}) . \qedhere
    \end{align*}
\end{proof}

\subsection{Analyzing the reduction via a martingale approach}

In order to obtain a bound on $\rho_{\cC}$ we will use two bounds: the recurrence obtained in  \cref{prop:recurrence} along with the sub-optimal bound $$\rho_{\cC} \lesssim_{\Delta} |\cG|$$ which holds by \cref{lem:crude-LSI} and standard comparison arguments provided $1/\Delta^{1000}\leq \alpha \leq 1/\Delta^2$, for instance.

We will interpret the bound in \cref{prop:recurrence} via the following stochastic process: Initialize $\cG_0 = \cG$ and $k_0 = k$. Assume that at step $i$, we have a configuration on $\mathcal{G}_i$ with parameter $k_i$ and $\alpha_i = k_i/\abs{\mathcal{G}_i}$. In every step, an adversary selects a pair $I\sim J$, then we pick uniformly at random a component $G=G_i$ of $\mathcal{G}_i$ with $I_G=J_G$, and set the subsequent configuration to be \begin{equation}\label{eq:config-update}
    \cG_{i+1} = \cG_i - G_i\,, \qquad k_{i+1} = k_i - |I_G| (= k_i - |J_G|)\,.
\end{equation} We will keep track of $\alpha=\alpha_i = k_i/ |\cG_i|$ along the process.  We denote $\mathcal{C}_i = (\mathcal{G}_i,k_i)$ the configuration at time $i$ and $\rho_i = \rho_{\cC_i}$.  We note that this is equivalent to iterating \cref{prop:recurrence} since  if $I\sim J$ and $I \neq J$ then then $I$ and $J$ differ in exactly two components, and $I$ and $J$ are the same in the remaining $m-2$ components.

Our goal will be to run the process until time $T = m - \log^4 n$, and we hope that $\alpha_j$ does not move too much.  With this in mind, define the stopping time $$\tau  = \min_t\{ |\alpha_t - \alpha_0| \geq 1/\sqrt{\log{n}}\}\,.  $$

Via a martingale argument, we will prove 
\begin{lemma}\label{lem:stopping-time-rare}
    In the above notation, we have $$\P\left(\tau < T \right) \lesssim_\Delta n^{-100}\,.$$
\end{lemma}

We now record the iterative bound that follows on $\rho_{\cC}$ provided that the occupancy ratio always satisfies the assumptions of \cref{prop:recurrence}.  In particular this is satisfied if  $s < \tau$.
\begin{lemma}\label{lem:iterative-bound-before-stopping}
    For all $s < \tau$ we have 
    $$\rho_{\mathcal{C}} \lesssim_{\Delta} \left(\frac{m-s}{m}\right)^2\frac{|\mathcal{G}_0|}{|\mathcal{G}_s|}\rho_s + \sum_{i\le s}\left(\frac{m-i}{m}\right)^2 \cdot \frac{|\cG_0|}{|\cG_i|} \cdot \frac{1}{m-i}\,.$$
\end{lemma}
\begin{proof}
    This follows from iterating the bound of \cref{prop:recurrence}, using the fact that 
    \[\prod_{i=1}^{s} \frac{m-i-2}{m-i} = \Omega_{\Delta}\left(\left(\frac{m-s}{m}\right)^2\right)\,. \qedhere 
    \]
\end{proof}

Assuming \cref{lem:stopping-time-rare}, \cref{thm:lee-yau} follows quickly.

\begin{proof}[Proof of Theorem \ref{thm:lee-yau}]
    Apply the recursive bound in  \cref{prop:recurrence} up to time $R = \min\{\tau,T\}$.  In the event that $\tau \geq T$ we apply  \cref{lem:iterative-bound-before-stopping} to see 
    \begin{align*}
    \rho_{\cC} \lesssim_{\Delta} \left(\frac{m - T}{m}\right)^2 \frac{|\cG_0|}{|\cG_T|} \rho_{\cC_T} + \sum_{i \leq T} \left(\frac{m - i}{m} \right)^2\frac{|\cG_0|}{|\cG_i|}\frac{1}{m - i}\,.
\end{align*}
The second term is $O_{\Delta}(1)$, noting that under the assumption that $G$ is $\Delta$-good, ${|\cG_0|}/{|\cG_i|} \le n/(m-i) =  O_{\Delta}(m/(m-i))$.  For the former, recall that $$\rho_{\cC_T} \lesssim_{\Delta} |\cG_T| $$ 
 and so we have $$\rho_{\cC} \lesssim_{\Delta} \frac{\log^8n}{m^2}|\cG_0| + 1\,. $$
 Noting that $m = \Omega_{\Delta}(n)$ by $\Delta$-goodness completes the bound on the event that $\tau > T$.  To bound $\rho_{\cC}$ on the event that $\tau < T$ we simply bound $$\E \rho_{\cC} 1\{\tau < T\} \lesssim n \P(\tau < T) = o_n(1),$$
 where we used  \cref{lem:crude-LSI} for the first bound and  \cref{lem:iterative-bound-before-stopping} for the second.  
\end{proof}

All that remains to prove is  \cref{lem:stopping-time-rare}.  This will follow from tracking the random sequence $\alpha_j$ via a martingale argument.

\begin{proof}[Proof of  \cref{lem:stopping-time-rare}]
    
Denote $\mathcal{F}_i$ the $\sigma$-algebra generated by the process up to time $i$.  We want to control the increments $\alpha_{i+1} - \alpha_i$; for simplicity of notation, in the following argument we denote $\mathcal{G}=\mathcal{G}_i$ and $G=G_i$. 
Observe that 
\[
\alpha_{i+1}-\alpha_i = \frac{k_i-\abs{I_G}}{|\mathcal{G}|-|G|}-\frac{k_i}{|\mathcal{G}|} = \frac{k_i|G| - |I_G| |\cG|}{(|\cG| - |G|)|\cG|}\,.
\]
Since $\sum_{G} |G| = |\cG|$ and $\sum_{G} |I_G| = k_i$ we have that $\sum_{G:I_G=J_G}|G| = |\cG|+O_{\Delta}(\max_G |G|)$, and $\sum_{G:I_G=J_G} |I_G| = k_i+O_{\Delta}(\max_G |G|)$, and hence

\begin{align*}
\E[\alpha_{i+1}-\alpha_i|\mathcal{F}_i] %
&= \frac{1}{m_i-2}\sum_{G:I_G=J_G}\left(\frac{k_i|G|-|I_G||\mathcal{G}|}{(|\mathcal{G}|-|G|)|\mathcal{G}|}\right) \\
&= \frac{1}{m_i-2} \sum_{G:I_G=J_G} \frac{k_i|G||\mathcal{G}| - k_i|G|(|\mathcal{G}|-|G|) - |I_G||\mathcal{G}|^2 + |I_G||\mathcal{G}|(|\mathcal{G}|-|G|)}{|\mathcal{G}|^2 (|\mathcal{G}|-|G|)} + O_{\Delta}\left(\frac{\max_G |G|}{m_i|\cG|}\right)\\
&= \frac{1}{m_i-2}\sum_{G:I_G=J_G}\frac{k_i|G|^2}{|\mathcal{G}|^2(|\mathcal{G}|-|G|)} - \frac{1}{m_i-2}\sum_{G:I_G=J_G} \frac{|I_G||G|}{|\mathcal{G}|(|\mathcal{G}|-|G|)}+ O_{\Delta}\left(\frac{\max_G |G|}{m_i|\cG|}\right). 
\end{align*}

Thus for $i \leq m - \log^4 n$ we have
\begin{align*}
    |\E[\alpha_{i+1}-\alpha_i|\mathcal{F}_i]|\le \frac{1}{m_i-2}\frac{k}{|\mathcal{G}|^3} \sum_{G}|G|^2 + \frac{1}{m_i-2}\cdot \frac{1}{|\mathcal{G}|^2} \sum_{G}|G|^2 +  O_{\Delta}\left(\frac{\max_G |G|}{m|\cG|}\right)\le O_{\Delta}\left( \frac{\log n}{m|\cG|}\right)\,.
\end{align*}

On the other hand,
\begin{align*}
    |\alpha_{i+1}-\alpha_i|\le \left|\frac{k_i-|I_G|}{|\mathcal{G}|-|G|} - \frac{k_i}{|\mathcal{G}|}\right|\le 2\frac{||\mathcal{G}||I_G| - k_i|G||}{|\mathcal{G}|^2}\le 4\frac{|G|}{|\mathcal{G}|}\,.
\end{align*}
Moreover,
\begin{align*}
\E[(\alpha_{i+1} - \alpha_i)^2 | \mathcal{F}_i]
&= \frac{1}{m_i-2}\sum_{G:I_G=J_G}\left(\frac{k_i|G|-|I_G||\mathcal{G}|}{(|\mathcal{G}|-|G|)|\mathcal{G}|}\right)^2\\
&\leq \frac{4}{m_i-2}\sum_{G \in \mathcal{G}}\frac{|G|^2}{|\mathcal{G}|^2} 
\end{align*}

We now construct the martingale $\bar{\alpha}_i$ by $$\bar{\alpha}_{i+1} = \bar{\alpha}_i + (\alpha_{i+1}-\alpha_i)-\E[(\alpha_{i+1}-\alpha_i)|\mathcal{F}_i].$$
The quadratic variation of $\bar{\alpha}$ satisfies $$\langle \bar{\alpha} \rangle_i = \sum_{j \leq i} \E( (\bar{\alpha}_j - \bar{\alpha}_{j-1})^2  \,|\,\mathcal{F}_{j-1})  \lesssim \sum_{j\le i} \left(\frac{1}{m_i}\sum_{G\in \mathcal{G}_i}|G|^2/|\mathcal{G}_i|^2\right) + \sum_{j\leq i}\log^2{n}/(m_i^2|\mathcal{G}_i|^2).$$
Consider $s\le m-\log^4 n$.  Note that since $|G| = O_{\Delta}(\log n)$ and $|\cG_i| \geq \log^4 n$ we have that \begin{align*}
    \left|\sum_{i\le s} \E[(\alpha_{i+1}-\alpha_i)|\mathcal{F}_i] \right| \lesssim_\Delta \sum_{i \leq s} \frac{1}{m - i} \frac{\log n}{|\cG_i|} \lesssim  \log n \sum_{i \leq s} \frac{1}{(m-i)^2} \leq \frac{\log n}{(m-s)} \lesssim \frac{1}{\log^3 n}\,.%
\end{align*} %
The same argument provides the bounds \begin{align*}
    \langle \bar{\alpha}\rangle_s \lesssim_\Delta \frac{1}{\log^3 n} \quad \text{ and } \quad \sup |\bar{\alpha}_{i+1}-\bar{\alpha}_i| \lesssim_\Delta \frac{1}{\log^3 n}\,.
\end{align*}
We will use the martingale Bernstein inequality by Freedman \cite{Fre75}. %
For $s \le m- \log^4 n $ we have the bound $$\P\left(|\bar{\alpha}_s-\bar{\alpha}_0|>x \right) \leq 2 \exp\left(- \Omega_{\Delta}\left(\frac{x^2 \log^3 n}{(1 + x)} \right) \right)\,.$$  
Applying this with $x = 1/\sqrt{\log{n}}$, we may sum over $s \leq m - \log^4 n$ to bound  $\P(\tau < T) \leq n^{-100}$ for $n$ sufficiently large in $\Delta$.  
\end{proof}

\subsection{Optimal decomposition using coupling: Proof of \cref{prop:opt-decomp}} \label{sec:opt-decomp}

Throughout, let $\mathcal{G}$ be as in the statement of \cref{prop:opt-decomp}. Let $\{G\}$ denote the connected components of $\mathcal{G}$. Throughout, we will fix a connected component $G$ of $\mc{G}$; an independent set $I_G$ of $G$; and a distinguished vertex $u \in I_G \subseteq G$. We also fix a non-negative function $f$. %

We now establish some notation. Let $\mu_k$ denote the uniform distribution on independent sets of $\mc{G}$ of size $k$. Let $\mu_1$ denote the distribution of $I' \sim \mu_k$ conditioned on the restriction of $I'$ to $G$ agreeing exactly with $I_G$. Let $\mu_2$ denote the distribution defined similarly, except we want agreement with $I_G \setminus u$. Finally, let $\mu$ denote the distribution of $I' \sim \mu_k$ conditioned on the restriction to $G\setminus u$ agreeing with $I_G$ (and hence, also with $I_G \setminus u$); in particular, the conditioning involved in the definition of $\mu$ places no restriction on the occupancy status of $u$. For a subset $S$ of vertices of $\mc{G}$, we let $N_+(S)$ denote the set consisting of $S$ and all neighbors of $S$ and we let $N_-(S) = N_+(S)\setminus S$. We will also let $C$ denote an absolute constant, which is sufficiently large to make an inequality at the end of this subsection hold.   

 Note that $\mu$ is a convex combination of $\mu_1$ and $\mu_2$ and so we can write $\mu = \theta \mu_1 + (1 - \theta) \mu_2$.  If $1/(C\Delta^{1000}) \leq k/n \leq 1/(C\Delta^8)$, there is some $\eps = \eps(\Delta) > 0$ so that $\theta \in [\eps,1 - \eps]$. 
As such we have 
\begin{equation}\label{eq:deriv-bounded}
\frac{d\mu_j}{d\mu} = O_\Delta(1)\end{equation} for $j \in \{1,2\}$. In particular, for a non-negative function $f$, we have $\E_{\mu_j}[f] = O_{\Delta}(\E_{\mu}[f])$.  Additionally, we note that by definition we have \begin{equation} \label{eq:mu1-mu2-fG}
    \E_{\mu_1} f = f_G(I_G), \qquad \E_{\mu_2} f = f_G(I_G \setminus u)\,.
\end{equation}
Thus, in order to prove \cref{prop:opt-decomp}, we want to derive a suitable upper bound on $\left(\sqrt{\E_{\mu_1} f} - \sqrt{ \E_{\mu_2} f} \right)^2$. We begin with the following preliminary lemma, which shows that it suffices to obtain an upper bound with one of the $\mu_i$s replaced by $\mu$.

\begin{lemma}\label{lem:mu2-mu}
    For $1/(C\Delta^{1000})  \leq k/n \leq 1/(C\Delta^6)$, we have
    \begin{equation*}
    \left(\sqrt{\E_{\mu_1} f} - \sqrt{ \E_{\mu_2} f} \right)^2 \lesssim_{\Delta} \left(\sqrt{\E_{\mu_1} f} - \sqrt{ \E_{\mu} f} \right)^2\,.
\end{equation*}
\end{lemma}
\begin{proof}
    Bound 
 \begin{align*}
        \left(\sqrt{\E_{\mu_1} f} - \sqrt{ \E_{\mu_2} f} \right)^2 &\leq \frac{(\E_{\mu_1} f - \E_{\mu_2} f)^2}{(\E_{\mu_1} f + \E_{\mu_2} f)}\,.
    \end{align*}
    And similarly bound \begin{align*}
         \left(\sqrt{\E_{\mu_1} f} - \sqrt{ \E_{\mu} f} \right)^2 &\geq \frac{1}{2}\frac{(\E_{\mu_1} f - \E_{\mu} f)^2}{(\E_{\mu_1} f + \E_{\mu} f)}  \geq \frac{1}{2} \frac{(1-\theta)^2(\E_{\mu_1} f - \E_{\mu_2} f)^2 }{\E_{\mu_1} f + \E_{\mu_2} f}\,. 
    \end{align*}
    Using \cref{eq:deriv-bounded} completes the proof.
\end{proof}
From now on, for notational simplicity, we will denote $\mu_1$ by $\nu$. We introduce two versions of the modified down-up walk with slight variations so that they are reversible with respect to $\mu$ and $\nu$.  The transition $P_\mu$ takes an independent set that on $G \setminus u$ agrees with $I_G$; it then picks two vertices $a,b \in \cG$ uniformly at random and moves to $P_{a,b}(I)$ provided that $a,b$ are in different connected components and $P_{a,b}(I)$ agrees with $I_G$ on $G \setminus u$. The transition $P_{\nu}$ is defined similarly, the only change being that the independent must agree with $I_G$ on $G$ throughout.

Let $(X_t)_{t \geq 0}$ evolve according to the transition $P_\mu$.  Then we may define 
\[ h(x) = \sum_{t=0}^{\infty} \E[f(X_t)-\E_\mu f | X_0 =x].\]
As shown in \cite{bresler2019stein,reinert2019approximating}, this is well-defined and satisfies the following Poisson equation:  $f = \E_\mu f + h - P_\mu h$. In particular, using that $\E_{\nu}[P_{\nu}(h) - h] = 0$, we have that
\[
\E_\nu f = \E_\nu [\E_\mu f + h - P_\mu h]  = \E_\mu f + \E_\nu [(P_\nu h - h) - (P_\mu h - h)].
\]
Hence, 
\begin{align}\label{eq:stein}
    \left(\sqrt{\E_\mu f}-\sqrt{\E_\nu f}\right)^2 
    \leq \frac{|\E_\mu f - \E_\nu f|^2}{\E_{\mu}f + \E_{\nu}f}
    \leq \frac{\left|\E_\nu (P_\mu h - h) - \E_\nu (P_\nu h - h)\right|^2}{\E_\mu f} 
\end{align}

Consider a sample $X_0 \sim \nu$ and couple $(Y_0,Y_0')$ so that $(X_0,Y_0)$ is a transition from $P_\mu$ and $(X_0,Y_0')$ is a transition from $P_\nu$.  Concretely, pick two vertices $a,b \in \mc{G}$ uniformly at random.  If either $a$ or $b$ is equal to $u$, then we set $Y_0' = X_0$ and set $Y_0 = P_{a,b}(X_0)$.  If neither $a$ nor $b$ is equal to $u$, we have two cases:  if at least one of $a$ or $b$ is in $G$, then $Y_0 = Y_0' = X_0$; otherwise we set $Y_0 = Y_0' = P_{a,b}(X_0)$.  In sum, the only case in which $Y_0 \neq Y_0'$ is in the event that either $a$ or $b$ is equal to $u$, which occurs with probability at most $2/n$. Moreover, observe that the marginal distributions of $Y_0$ and $Y_0'$ in $(Y_0, Y_0')$ conditioned on the event $Y_0 \neq Y_0'$ coincide with $\mu_1 (= \nu)$ and $\mu_2$. %

By construction,
\begin{align*}
    \E_{\nu}(P_\mu h - h) - \E_{\nu}(P_{\nu}h-h) = \E \E_{(Y_0, Y_0')} \sum_{t = 0}^{\infty} [f(Y_t) - f(Y_t')],%
\end{align*}
where $\{(Y_t, Y_t')\}_{t\geq 0}$ is a contractive coupling of two trajectories evolving according to $P_{\mu}$, started from $(Y_0, Y_0')$, and where $(Y_0, Y_0')$ is distributed as in the previous paragraph. The outer expectation is taken over the randomness of the trajectory.
We will later describe a single step of the contractive coupling in the proof of \cref{lem:bound-H-gamma}, but for now we only need that once the two walks couple, they remain equal thereafter, and that the coupling time has exponential tails.  
Hence,
\begin{align}
    \frac{\left|\E_\nu (P_\mu h - h) - \E_\nu (P_\nu h - h)\right|^2}{\E_\mu f}
    &= \frac{\left|\E \left[\sum_t (f(Y_t)-f(Y_t'))\right]\right|^2}{\E_\mu f} \nonumber \\
    &= \frac{\left|\P(Y_0 \neq Y_0')  \E[\sum_t (f(Y_t)-f(Y_t'))\,|\,Y_0 \neq Y_0']\right|^2}{\E_\mu f } \nonumber  \\
    &\lesssim \frac{1}{n^2}\frac{\left| \E  \E_{(Y_0,Y_0') \sim \gamma}[\sum_t (f(Y_t)-f(Y_t'))]\right|^2}{\E_\mu f}, \label{eq:dirichlet-to-gamma}
\end{align}
where $\gamma$ denotes the following distribution: first choose $v \notin G$ uniformly at random, then choose $I'$ from the uniform distribution on independent sets of size $(k-I_G)$ in $\mc{G}\setminus (G \cup N_+(v))$, and finally, let $Y_0 = I_G \cup I'$, $Y_0' = (I_G\setminus u) \cup I' \cup v$.  

Let $\tau$ be the coupling time of $(Y_t,Y_t')_{t\geq 0}$ and let $c > 0$ be a sufficiently small constant to be chosen later. Let 
\[\zeta = \E \E_{(Y_0, Y_0') \sim \gamma}[f(Y_t) + f(Y_t')].\]
From \cref{eq:deriv-bounded}, it follows that $\zeta = O_{\Delta}(\E_{\mu}[f])$. Therefore, we have
\begin{align}
    \frac{\left|\E_\nu (P_\mu h - h) - \E_\nu (P_\nu h - h)\right|^2}{\E_\mu f} 
    &\lesssim \frac{1}{n^2}  \frac{\left|\E\E_{(Y_0,Y_0')\sim \gamma} [\sum_t (f(Y_t)-f(Y_t'))]\right|^2}{\E_\mu f} \nonumber \\
    &\leq \frac{1}{n^2} \frac{1}{\E_\mu f}\left(\sum_{t} e^{-ct/n}\zeta\right) \cdot \sum_t e^{ct/n}\frac{(\E\E_{(Y_0,Y_0')\sim \gamma}[f(Y_t)-f(Y_t')])^2}{\zeta} \nonumber \\
    &\lesssim_{\Delta} \frac{1}{n}\sum_{t\le \tau}\exp(ct/n)\frac{(\E\E_{(Y_0,Y_0')\sim \gamma}[f(Y_t)-f(Y_t')])^2}{\E\E_{(Y_0, Y_0')\sim \gamma}[f(Y_t)+f(Y_t')]} \nonumber \\
    &\leq \frac{1}{n}\E\E_{(Y_0,Y_0')\sim \gamma} \sum_{t\le \tau}\exp(ct/n) \left|\sqrt{f(Y_t)} - \sqrt{f(Y_t')}\right|^2, \label{eq:bound-by-H}
\end{align}
where the second and last inequalities are Cauchy-Schwarz. 
For a distribution $\Gamma$ on $(Y_0, Y_0')$, let 
\begin{align*}
    H(\Gamma) :=  \E\E_{(Y_0,Y_0') \sim \Gamma} \sum_{t\le \tau} \exp(ct/n)|\sqrt{f(Y_t)} - \sqrt{f(Y_t')}|^2. 
\end{align*}

Our last step will be to bound $H(\gamma)$.  
\begin{lemma} \label{lem:bound-H-gamma}
There are absolute constants $c$ and $C$ such that for $1/(C\Delta^{1000}) \leq k/n \leq 1/(C\Delta^8)$,
   \[ H(\gamma) \leq 2n \cdot \E_{(Y_0, Y_0')\sim \gamma}\left[\sqrt{f(Y_0)} - \sqrt{f(Y_0')}\right]^2 \]
\end{lemma}
\begin{proof}
    We will bound $H(\gamma)$ using a recursion derived via first-step analysis. To this end, we describe the first step of the coupling of the trajectories, started from independent sets $(Y_0, Y_0')$ in the support of $\gamma$. By definition, $Y_0 \neq Y_0'$ with $u \in Y_0$ and $Y_0' = P_{u,v}(Y_0)$ for some $v\notin G$. We pick two vertices $a,b \in \mc{G}$ uniformly at random. If $a$ is equal to $u$ or $v$, then we set $Y_1 = P_{u,b}(Y_0), Y_1' = P_{v,b}(Y_0')$. Else, if $b$ is equal to $u$ or $v$, then we set $Y_1 = P_{a,u}(Y_0), Y_1' = P_{a,v}(Y_0')$. If neither $a$ nor $b$ equals $u$ or $v$, we have two cases: if at least one of $a$ or $b$ is in $G$, we set $Y_1 = Y_0, Y_1' = Y_0'$, else we set $Y_1 = P_{a,b}(Y_0)$, $Y_1' = P_{a,b}(Y_0')$. 

    Note that there are three possibilities: (i) $Y_1 = Y_1'$, (ii) $Y_1' = P_{u,v}(Y_1)$, (iii) $Y_1 = P_{w,w'}(Y_0), Y_1' = Y_0'$ for some $w \in Y_0 \setminus (G \cup v)$ and $w' \in N_-(v)$. Let $\rho$ denote the probability of the third case. The probability of the first case is at least $1/n$, and hence, the probability of the second case is at most $(1-\rho - 1/n)$. Note that $\rho \leq 2k/n\cdot \Delta/n$. Also, observe that the distribution of $(Y_1, Y_1')$ conditioned on the event that $Y_1' = P_{u,v}(Y_1)$ coincides with $\gamma$. Therefore, we may write:
    \begin{align*}
        H(\gamma) \leq \E_{(Y_0, Y_0')\sim \gamma}[\sqrt{f(Y_0)} - \sqrt{f(Y_0')}]^2 + e^{c/n}(1-\rho-1/n)H(\gamma) + e^{c/n}\rho\cdot H((Y_1,Y_1')|\text{(iii)}).
    \end{align*}
    
    We now bound the third term. With notation as above, we consider the following distribution on `paths' between $Y_1$ and $Y_1' (=Y_0')$. Let $\delta$ be a uniformly random vertex outside the distance two neighborhood of $Y_0 \cup Y_0'$ and consider the path $Y_1 \to Z_0 \to Z_1 \to Z_2 \to Y_0 \to Y_1'$ where $Z_0 = P_{u,\delta}(Y_1)$, $Z_1 = P_{u,w'}(Z_0)$, $Z_2 = P_{u,w}(Z_1)$ (so that $Y_0 = P_{u,\delta}(Z_2)$), and as before, $Y_1' = Y_0' = P_{u,v}(Y_0)$. Observe that, over the randomness of $\gamma$, $w' \in N_-(v)$ has relative density $O(\Delta)$ with respect to the uniform distribution on vertices in $\mc{G}\setminus G$. Moreover, by \cref{eqn:prob-hitting}, over the randomness of $\gamma$, the distribution of $w'$ has relative density $O(1)$ with respect to the uniform distribution on vertices in $\mc{G}\setminus G$. It follows that each of the five distributions $(Y_1, Z_0), (Z_1, Z_0), (Z_1, Z_2), (Y_0, Z_2), (Y_0, Y_1')$ has density $O(\Delta)$ with respect to $\gamma$. Therefore, by interpolating along this path and using Cauchy-Schwarz, we may bound
    \begin{align*}
        H((Y_1,Y_1')|\text{(iii)}) \lesssim \Delta \cdot H(\gamma);
    \end{align*}
    thus
    \begin{align*}
        H(\gamma) &\leq \E_{(Y_0, Y_0')\sim \gamma}[\sqrt{f(Y_0)} - \sqrt{f(Y_0')}]^2 + e^{c/n}(1-1/n+O(\Delta \rho))H(\gamma) \\ 
        &\leq \E_{(Y_0, Y_0')\sim \gamma}[\sqrt{f(Y_0)} - \sqrt{f(Y_0')}]^2 + (1-1/2n)H(\gamma),
    \end{align*}
    provided we pick $c = 1/1000$ and $k/n \leq \Delta^{-8}/C$ for a sufficiently large absolute constant $C$. From this, we immediately get the required estimate
    \[H(\gamma) \leq 2n \cdot \E_{(Y_0, Y_0')\sim \gamma}\left[\sqrt{f(Y_0)} - \sqrt{f(Y_0')}\right]^2. \qedhere\]
\end{proof}

\begin{proof}[Proof of \cref{prop:opt-decomp}]
    Combining \cref{eq:mu1-mu2-fG}, \cref{lem:mu2-mu}, \cref{eq:stein}, \cref{eq:bound-by-H} and \cref{lem:bound-H-gamma} shows that $$(\sqrt{f_G(I_G)} - \sqrt{f_G(I_G \setminus u)})^2 \lesssim_{\Delta} \E_{(Y_0,Y_0') \sim \gamma}\left[\sqrt{f(Y_0)} - \sqrt{f(Y_0')}\right]^2\,. $$ Recalling the definition of $\gamma$ shows \[ \E_{(Y_0,Y_0') \sim \gamma}\left[\sqrt{f(Y_0)} - \sqrt{f(Y_0')}\right]^2 \lesssim_{\Delta} \E_{v} \E_{I_{-G}|I_G}  \left[\left(\sqrt{f(I)}-\sqrt{f(P_{u,v}(I))}\right)^2\right]\,. \qedhere\]
\end{proof}

\section*{Acknowledgments}
M.M. is supported in part by NSF grant DMS-2137623. H.T.P. is supported by a Two Sigma Fellowship. 
\bibliographystyle{abbrv}
\bibliography{main.bib}

\end{document}